\documentclass[journal]{IEEEtran}
\usepackage[latin1]{inputenc}
\usepackage{amsmath}
\usepackage{graphicx}
\usepackage{amssymb}
\usepackage{xcolor}
\makeatletter

\newtheorem{lemma}{Lemma}

\providecommand{\LyX}{L\kern-.1667em\lower.25em\hbox{Y}\kern-.125emX\@}

\def\BibTeX{{\rm B\kern-.05em{\sc i\kern-.025em b}\kern-.08em
    T\kern-.1667em\lower.7ex\hbox{E}\kern-.125emX}}

\makeatother
\begin{document}
\def\ZZ{{\mathbb Z}}
\def\RR{{\Bbb R}}
\def\NN{{\mathbb N}}
\def\CC{{\mathbb C}}

\title{Time-Delay Compensators for Linear Systems with Delayed Output Measurements}

\author{H. Trinh, P. T. Nam, T. N. Nguyen\thanks{{\em}}
\thanks{Hieu Trinh is with the School of Engineering, Deakin University, Geelong, 3217, Australia (e-mail:
hieu.trinh@deakin.edu.au).}
\thanks{Phan Thanh Nam and Tran Ngoc Nguyen are with the Department of Mathematics, Quy Nhon University, Vietnam (email: phanthanhnam@qnu.edu.vn, tranngocnguyen@qnu.edu.vn)}}

\maketitle \maketitle \maketitle \thispagestyle{plain}
\pagestyle{plain}

\begin{abstract}
This paper provides a comprehensive framework for designing functional observers for linear systems subject to delayed output measurements.  Moving beyond traditional methodologies, the proposed observer generates an estimate $\hat{z}(t)$ that predicts the current state functional $z(t)=Fx(t)$ using delayed data. By neutralizing sensing latency, the observer serves as a potent time-delay compensator, effectively expanding the practical utility of functional observer theory. The proposed observer architecture offers greater robustness and versatility than traditional Luenberger-type observers by leveraging multiple delayed components to preserve accuracy despite latency. A key contribution of this work is a novel method for extending the maximum allowable measurement delay while maintaining the asymptotic stability of the estimation-error system. Existence conditions are established together with constructive synthesis procedures. Extensive numerical examples are given to illustrate the proposed theory.

\end{abstract}

\section{System Description and Problem Statement}
\label{ch1sec1}
We consider the following linear system
\begin{align}
	\label{c1.1}
	\dot{x}(t)=Ax(t)+Bu(t),
\end{align}
where $x(t)\in \mathbb{R}^n$ is the state vector, $u(t)\in \mathbb{R}^r$ is the control input vector. Matrices
$A\in\mathbb{R}^{n\times n}$ and $B\in\mathbb{R}^{n\times r}$ are constant. 

Due to sensing or communication constraints, the measured output vector is available with delay and is modeled as
\begin{align}
	\label{c1.2}
	y(t) = C_\tau x(t-\tau),
\end{align}
where  $\tau>0$ is the time delay, $y(t)\in\mathbb{R}^{p}$ with $0<p\le n$, and
$C_\tau\in\mathbb{R}^{p\times n}$ is a constant matrix of full row rank.

We are interested in estimating the following functional
\begin{equation}\label{c1.3}
	z(t)=Fx(t),
\end{equation}
where $z(t)\in\mathbb{R}^m$ is a linear function of the instantaneous state vector, $x(t)$, and $F\in\mathbb{R}^{m\times n}$ is full row rank. The full row rank assumption ensures that the rows of $F$ are linearly independent, so that the components of $z(t)$ represent independent linear functionals of the state vector. This entails no loss of generality, since any linearly dependent rows can be removed without altering the functional subspace spanned by $F$.

In the case $F=I_n$ and thus $z(t)=x(t)$, the estimation problem reduces to recovering the instantaneous state vector $x(t)$ from delayed measurements $y(t)=C_{\tau}x(t-\tau)$. This formulation generalizes the Luenberger state observer \cite{luen1} under the condition of delayed outputs.

To estimate functional (\ref{c1.3}), we first consider the following observer with an internal delay
\begin{align}
	\label{c1.4a}
	\hat{z}(t)&=w(t)+My(t),\\
	\label{c1.4b}
	\dot{w}(t)&=Nw(t)+N_{\tau}w(t-\tau)+Gy(t)+G_{\tau}y(t-\tau)\nonumber\\&+Ju(t)+J_{\tau}u(t-\tau),
\end{align} 
with $w(\theta)=\rho(\theta),\ \forall \theta \in[-\tau,0]$ is the initial condition,  $w(t) \in \mathbb{R}^m$, and $\hat{z}(t)$ is the estimate of $z(t)$. The matrices $M$, $N$, $N_{\tau}$, $G$, $G_{\tau}$, $J$, and $J_{\tau}$ are to be determined so that
$\hat{z}(t)\to z(t)$ asymptotically.

\textit{\textbf{Remark 1:}} The observer described in equations (\ref{c1.4a})-(\ref{c1.4b}) incorporates three distinct delayed components: An internal delay term $N_{\tau}w(t-\tau)$, delayed output measurements $G_{\tau}y(t-\tau)$ and delayed control inputs $J_{\tau}u(t-\tau)$. These terms are essential for effectively compensating for time delays within the output measurements. It is worth noting that this architecture is more robust and versatile than classical Luenberger state observers \cite{luen1} which assumes a delay-free output vector where $\tau=0$.

Figure (\ref{figc1.1}) shows a block diagram of the implementation of observer (\ref{c1.4a})-(\ref{c1.4b}) using delayed output vector $y(t)=C_{\tau}x(t-\tau)$. Observe that the output of the observer produces $\hat{z}(t)$ which predicts the current value of $Fx(t)$ using delayed measurements. Hence, the observer effectively compensates for the measurement delay
and may therefore be interpreted as a \textit{time-delay compensator}.
\begin{figure}[!h]
	\centering
	\includegraphics[width=\linewidth]{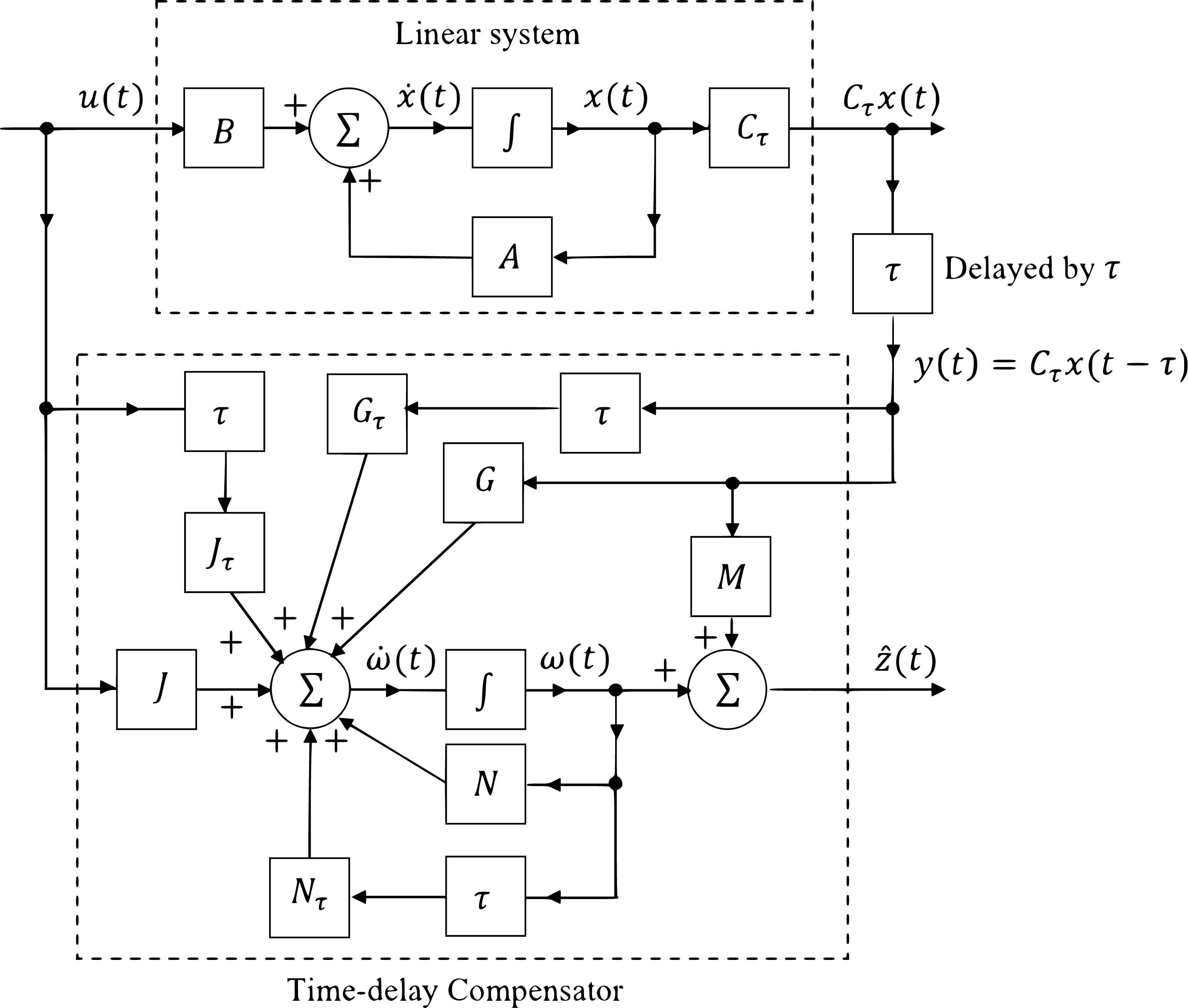}
	\caption{Schematic diagram of an observer acting as a time-delay compensator}
	\label{figc1.1}
\end{figure}

\textit{\textbf{Notation:}}
For a matrix $F$, $F^{\mathsf T}$ denotes its transpose, $F^{-}$ denotes a generalized inverse satisfying $FF^{-}F=F$, and $\mathrm{rank}(F)$ denotes the rank of $F$. The symbols $\mathbb R$ and $\mathbb C$ denote the sets of real and complex numbers.
For $\lambda\in\mathbb C$, $\Re(\lambda)$ denotes its real part. For brevity of notation, $\mathrm{sym}(F):=F+F^{\mathsf T}$.

A matrix $P \in\mathbb{R}^{n\times n}$ is called positive definite 
($P>0$) if $P=P^{\mathsf T}$ and $x^{\mathsf T}Px>0$ for all 
$x\neq \mathbf 0$. For a symmetric matrix $P=\begin{pmatrix} P_{11} &P_{12}\\ \star &P_{22}\end{pmatrix}$,
the symbol $\star$
denotes the transpose of the corresponding upper-right block.

For a square matrix $F$, $\sigma(F)$ denotes the set of all eigenvalues of $F$, and
\[
\sigma_{\min}(F)
=
\min\{\Re(\lambda):\lambda\in\sigma(F)\}.
\]

The identity matrix is denoted by $I$ when its dimension is clear 
from the context; otherwise it is specified using a subscript. 
The symbol $\mathbf 0$ denotes a zero matrix or zero vector; its 
dimension is inferred from the context and specified by a 
subscript only when necessary. The scalar zero is denoted by $0$.

For $G\in\mathbb{R}^{k\times n}$, we denote by 
$\operatorname{row}(G)\subseteq\mathbb{R}^n$ the subspace spanned 
by the rows of $G$.

\section{Observer Existence Conditions}
\label{ch1sec2}
We derive existence conditions of observer (\ref{c1.4a})-(\ref{c1.4b}). Defining the estimation error vector $e(t)=\hat z(t)-z(t)$, the error dynamics are given by
\begin{align}
	\label{c1.5}
	\dot{e}(t)&=\dot{w}(t)+M\dot{y}(t)-F\dot{x}(t)\nonumber\\ \quad &=Ne(t)+N_{\tau}e(t-\tau)+\mathcal{C}_{1}u(t)+\mathcal{C}_{2}u(t-\tau)\nonumber\\&+\mathcal{C}_{3}x(t)+\mathcal{C}_{4}x(t-\tau)+\mathcal{C}_{5}x(t-2\tau),
\end{align}
where\\ 
$\mathcal{C}_1 =J-FB,\quad \mathcal{C}_2 = J_{\tau}+MC_{\tau}B, \quad\mathcal{C}_3 =NF-FA$, \quad
$\mathcal{C}_4 = N_{\tau}F+\bar{G}C_{\tau}+MC_{\tau}A$, \quad $\mathcal{C}_5 = \bar{G}_{\tau}C_{\tau},$ \quad $\bar{G}:=G-NM$, \quad $\bar{G}_{\tau}:=G_{\tau}-N_{\tau}M$, \quad and \quad
$\mathcal C
:=
\begin{pmatrix}
	\mathcal C_1 &
	\mathcal C_2 &
	\mathcal C_3 &
	\mathcal C_4 &
	\mathcal C_5
\end{pmatrix}.
$

The following theorem provides conditions for the existence of observer (\ref{c1.4a})-(\ref{c1.4b}).

\textit{\textbf{Theorem 1:}}
	\textit{Observer (\ref{c1.4a})-(\ref{c1.4b}) provides asymptotic estimation of the functional $z(t)=Fx(t)$ and yields
	estimation error dynamics that are decoupled from the plant state $x(\cdot)$ and the input $u(\cdot)$
	if $\mathcal C=\bf 0$ and the following delay-dependent error dynamics
	\begin{align}
		\label{c1.6}&\dot{e}(t)=Ne(t)+N_{\tau}e(t-\tau)
	\end{align}
	is asymptotically stable. In this case, the estimation error satisfies
	\[
	e(t)=\hat z(t)-z(t)\to {\bf 0} \quad \text{as} \quad t\to\infty
	\]
	for all admissible initial conditions and inputs $u(\cdot)$}.

\textit{\textbf{Proof:}} If $\mathcal C=\bf 0$, then \eqref{c1.5} reduces to \eqref{c1.6}, so the error dynamics are
	decoupled from $x(\cdot)$ and $u(\cdot)$. If, in addition, \eqref{c1.6} is asymptotically stable, then $e(t)\to \bf 0$ as $t\to\infty$ for all admissible initial conditions and inputs $u(\cdot)$. This completes the proof.

\section{Determination of the Observer Parameters}
\label{ch1sec3}
In this section, we derive observer parameters $M$, $N$, $N_{\tau}$, $G$, $G_{\tau}$, $J$ and $J_{\tau}$ to satisfy the conditions presented in Theorem 1.

Let us first consider $\mathcal C=\bf 0$. For this, $\mathcal C_1=\bf 0$ and $\mathcal C_2=\bf 0$ are satisfied by the following direct
choices of $J$ and $J_{\tau}$, respectively,
\begin{align}
	\label{c1.7}
	&J = FB, \quad J_{\tau}=-MC_{\tau}B.
\end{align}

Since $C_{\tau}$ is full row rank, $\mathcal C_5=\bf 0$ is satisfied iff $\bar{G}_{\tau}=\bf 0$, thus $G_{\tau}$ is obtained as follows
\begin{align}
	\label{c1.8}
	&G_{\tau}=N_{\tau}M.
\end{align}

$\mathcal C_3=\bf 0$ iff
\begin{equation}
	\label{c1.9}
	\mathrm{rank}\begin{pmatrix}FA\\ F\end{pmatrix}=\mathrm{rank}(F).
\end{equation}

Subject to the satisfaction of condition (\ref{c1.9}), we obtain
\begin{equation}
	\label{c1.10}
	N=FAF^-,
\end{equation}
where $F^{-}$ denotes a generalized inverse
satisfying $FF^{-}F=F$.

Therefore, the set of the eigenvalues of $N$ is fixed and belongs to a subset of the eigenvalues of $A$. For this reason, when $N$ is an unstable matrix, we need to employ an observer with an internal delay, $N_{\tau}w(t-\tau)$, to provide asymptotic stability of the error time-delay system (\ref{c1.6}).

Let $X\in \mathbb{R}^{m\times (m+2p)}$ be defined as follows
\begin{align}
	\label{c1.11}
	X:=\begin{pmatrix}N_{\tau} &\bar{G} &M\end{pmatrix}.
\end{align}
Then $\mathcal C_4=\bf 0$ can be expressed as follows
\begin{align}
	\label{c1.12}
	X\Theta=\mathbf{0}_{m\times n},
\end{align}
where
\begin{align}
	\label{c1.13}
	\Theta=\begin{pmatrix}F\\C_{\tau}\\C_{\tau}A\end{pmatrix}\in\mathbb{R}^{(m+2p)\times n}.
\end{align}

According to Lemma 2 (see the Appendix), a non-trivial solution $X\neq \mathbf 0$ exists if and only if there exists an orthogonal basis of the left null-space of $\Theta$.  The general solution to (\ref{c1.12}) is given by
\begin{equation}
	\label{c1.14}
	X=Z(I_{m+2p}-\Theta\Theta^{-}), \quad \Theta\Theta^{-}\neq I_{m+2p},
\end{equation}
where 
$Z\in\mathbb{R}^{m\times (m+2p)}$ is arbitrary.

Given that $(I_{m+2p}-\Theta\Theta^{+})\neq \mathbf 0$, there always exists $Z\neq \mathbf 0$ such that $X\neq \mathbf0$. The free matrix $Z$ reflects the inherent design freedom available in the observer construction. Therefore, we shall exploit this feature to solve the stabilization problem of finding $Z$ to ensure the asymptotic stability of the error time-delay system (\ref{c1.6}). 

Let $\mathcal{I}_1$ denote the first $m$ columns of $I_{m+2p}$. Right-multiplying (\ref{c1.14}) by $\mathcal{I}_1$, we obtain $N_{\tau}$ as follows
\begin{align}
	\label{c1.15a}
	N_{\tau}&=ZN_{\tau_2}, \\
	\label{c1.15b} N_{\tau_2}&=(I_{m+2p}-\Theta\Theta^{-})\mathcal{I}_1.
\end{align}

Therefore, the error time-delay system (\ref{c1.6}) is now expressed as follows
\begin{equation}
	\label{c1.16}
	\dot{e}(t)=Ne(t)+ZN_{\tau_2}e(t-\tau),
\end{equation}
where $N=FAF^-$.

Finally, we solve the stabilization problem by finding $Z$ to ensure the asymptotic stability of (\ref{c1.16}). For further context, the Appendix at the end of this paper provides enhanced stabilizability conditions for time-delay systems formulated as Linear Matrix Inequalities (LMIs). A key challenge arises when $N$ is an unstable matrix; this makes the stabilization problem more stringent. The LMI-based approach presented in the Appendix is specifically designed to accommodate such restrictive conditions. Detailed derivations and the resulting LMIs capable of handling these unstable matrices are provided in the Appendix.

Once asymptotic stability of (\ref{c1.16}) is ensured, $N_{\tau}$, $\bar G$ and $M$
are then obtained directly from the corresponding blocks of $X$ (see, \eqref{c1.11}), and the observer gains $G$, $G_{\tau}$ and $J_{\tau}$ are recovered from
\[ G=\bar{G}+NM, \quad G_{\tau}=N_{\tau}M, \quad J_{\tau}=-MC_{\tau}B. \]

Let us now look at some special cases arose from our solution to the functional observer design problem.

\subsection{Case 1:  $\mathrm{rank}\begin{pmatrix}C_{\tau}\\ C_{\tau}A\end{pmatrix}=n$, $N$ is not Hurwitz.}
\label{ch1subsec2}
This case provides further insights into the stabilization problem of the error time-delay system \eqref{c1.6} subject to the satisfaction of the condition $\mathrm{rank}\begin{pmatrix}C_{\tau}\\ C_{\tau}A\end{pmatrix}=n$. This condition requires that the observability index of the pair $(A,C_{\tau})$ is 2, which implies that the number of the delayed output measurements needs to be sufficiently large enough. 

Let us now express $\mathcal C_4=\bf 0$ as follows
\begin{align}
	\label{c1.19}
	\bar{X}\bar{\Theta}=-N_{\tau}F,
\end{align}
where $$\bar{X}:=\begin{pmatrix}\bar{G}  &M\end{pmatrix}, \quad \bar{\Theta}:= \begin{pmatrix}C_{\tau}\\C_{\tau}A\end{pmatrix}.$$
Since $\mathrm{rank}\begin{pmatrix}C_{\tau}\\ C_{\tau}A\end{pmatrix}=n$, by Lemma 1 (see the Appendix), a solution
$\bar{X}$ always exists for any given $N_{\tau}\neq \bf 0$, and is given by
\begin{align}
	\label{c1.20}
	\bar{X}=-N_{\tau}F\bar{\Theta}^-.
\end{align}

The above development shows that subject to $\mathrm{rank}\begin{pmatrix}C_{\tau}\\ C_{\tau}A\end{pmatrix}=n$, we can independently carry out the stabilization problem of finding $N_{\tau}$ such that the following time-delay error system
\[\dot{e}(t)=Ne(t)+N_{\tau}e(t-\tau),\]
$N=FAF^-$, is asymptotically stable. To solve this problem, we can use Lemma \ref{lem:stabilization-constant-delay} 
to find $N_{\tau}$ so that the above time-delay system is asymptotically stable for a given $\tau$. 
\subsection{Case 2:  $\mathrm{rank}\begin{pmatrix}C_{\tau}\\ C_{\tau}A\end{pmatrix}<n$, $N$ is not Hurwitz.}
\label{ch1subseccase2}
For the case where  condition $\mathrm{rank}\begin{pmatrix}C_{\tau}\\ C_{\tau}A\end{pmatrix}=n$ is not satisfied, i.e., $\mathrm{rank}\begin{pmatrix}C_{\tau}\\ C_{\tau}A\end{pmatrix} < n$. This situation arises, for example, where lesser number of delayed output measurements is available. In such a situation, we do not solve the stabilization problem of finding $N_{\tau}$ such that the error time-delay system (\ref{c1.6}) is asymptotically stable as discussed in Lemma  \ref{lem:stabilization-constant-delay}. Instead, we turn to Lemma \ref{lem:output-stabilization-constant-delay} to solve a modified stabilization problem of finding $\bar{Z}\in \mathbb{R}^{m\times v}$ such that the following error time-delay system
\begin{align}
	\label{c1.21}
	\dot{e}(t)=Ne(t)+\bar{Z}\bar{N}e(t-\tau),
\end{align} 
where $N=FAF^-$, $\bar{N}\in \mathbb{R}^{v\times m}$ is a known full-row rank matrix with $v<m$. As demonstrated in the Appendix, LMI of Lemma \ref{lem:output-stabilization-constant-delay} is more conservative than LMI of Lemma \ref{lem:stabilization-constant-delay}. This is as expected since this practical situation is due to having lesser number of available delayed output measurements.

\subsection{Case 3:  $\mathrm{rank}\begin{pmatrix}FA\\ F\end{pmatrix}\neq \mathrm{rank}(F)$.}
\label{ch1subseccase3}
When condition (\ref{c1.9}) is not satisfied, i.e., $\mathrm{rank}\begin{pmatrix}FA\\ F\end{pmatrix}> \mathrm{rank}(F)$, an observer (\ref{c1.4a})-(\ref{c1.4b}) of order $m$ does not exist. In this case, a higher-order functional observer (i.e., of order greater
than $m$) may still exist. Based on the augmentation approach \cite{fern1}, we increase the order of the designed observer by augmenting the orginal functional $z(t)$ with an extra functional, namely, $z_a(t)=Rx(t)$. Since the order of the observer equals the number of functionals being estimated, increasing the observer
order corresponds to enlarging the functional subspace.
Let
\[
z_a(t)=Rx(t), \ R\in\mathbb{R}^{(q-m)\times n}, \ q>m,
\]
be an additional functional such that
\[
\begin{pmatrix}
	F\\
	R
\end{pmatrix}
\]
has full row rank.
The augmented functional is then defined as
\[
z_{\mathrm{aug}}(t)
= \begin{pmatrix}z(t)\\z_a(t) \end{pmatrix}=
\begin{pmatrix}
	F\\
	R
\end{pmatrix}x(t)=\bar{F}x(t),
\]
where $\bar F:=\begin{pmatrix}F\\ R\end{pmatrix}\in\mathbb{R}^{q\times n}$.

Construction of matrix $R \in \mathbb{R}^{q-m}, \ q \geq m$ to satisfy 
\[\mathrm{rank}\begin{pmatrix}\bar{F}A\\ \bar{F}\end{pmatrix}= \mathrm{rank}(\bar{F})\]
can be systematically carried out. The reader can refer to \cite{trinhnam1} for a procedure to construct $R$. Briefly, let us define
\[
\bar F_0:=
\begin{pmatrix}
	F\\ FA\\ \vdots\\ FA^{n-1}
\end{pmatrix}
\]
and let
\[
q := \mathrm{rank} \left(\bar F_0 \right).
\]

Given that $\mathrm{rank} \left(\bar F_0 \right)=q$, we can select $q$ linearly independent rows of $\bar F_0$ that contains the rows of $F$. Then let us denote the resulting full--row--rank matrix as
\[
\bar F=
\begin{pmatrix}
	F\\ R
\end{pmatrix}\in\mathbb{R}^{q\times n},
\qquad q\ge m.
\]
Note that $\mathrm{row}(\bar F)=\mathrm{row}(\bar F_0)$. By the Cayley--Hamilton theorem, $FA^n$ is a linear combination of $F, FA, \dots, FA^{n-1}$.
Hence
$\mathrm{row}(\bar F A)\subseteq \mathrm{row}(\bar F)$ and therefore
\begin{align}
	\label{c1.rank}
	\mathrm{rank}\begin{pmatrix}\bar{F}A\\ \bar{F}\end{pmatrix}= \mathrm{rank}(\bar{F}).
\end{align}
Thus, when condition (\ref{c1.9}) is not satisfied, we can find $R$ based on the above, and upon replacing 
$F$ by $\bar F$, condition (\ref{c1.rank}) is satisfied.

We can then carry out the design of a $q$-order functional observer to estimate the functional $z_{\mathrm{aug}}(t)$. As a result, the estimation of $z(t)$ can be obtained as follows
\begin{align}
	\label{c1.extract}
	\hat{z}(t)=\begin{pmatrix}I_m &\mathbf{0}_{m\times (q-m)}\end{pmatrix}\hat{z}_{\mathrm{aug}}(t)=K\hat{z}_{\mathrm{aug}}(t),
\end{align}
where $K:=\begin{pmatrix}I_m &\mathbf{0}_{m\times (q-m)}\end{pmatrix}$.

\section{Numerical Examples}
\label{ch1sec4}
We now present some numerical examples to illustrate the proposed existence conditions and observer synthesis procedures described in Section \ref{ch1sec3}.

\textit{Numerical Example 1:} Let $A$, $B$ and $C_{\tau}$ be given as follows
$$A=\begin{pmatrix} 0.1 &1\\1 &-2 \end{pmatrix}, \ B=\begin{pmatrix}1\\2 \end{pmatrix}, \ C_{\tau}=I_2.$$
Matrix $A$ is unstable with $\sigma(A)=\{0.5, -2.4\}$. The delayed output measurement vector consists of the entire state vector being delayed by $\tau$, i.e.,
$$y(t)=x(t-\tau).$$
Our objective is to design an observer (or a time-delay compensator) to estimate instantaneously the state vector $x(t)$ using the delayed output vector $y(t)=x(t-\tau)$. Thus, we have $F=I_2$ and $z(t)=x(t)$.

With $C_{\tau}=I_2$ and $F=I_2$, condition \eqref{c1.9} is satisfied, but
\[
N=FAF^{-}=A
\]
is not Hurwitz. Therefore, we need to employ an observer with an internal delay, $N_{\tau}w(t-\tau)$, $N_{\tau}\neq \mathbf 0$, to provide asymptotic stability of the error time-delay system (\ref{c1.6}). 

For this example, we have $\mathrm{rank}\begin{pmatrix}C_{\tau}\\ C_{\tau}A\end{pmatrix}=n$, $N$ is not Hurwitz. Therefore, as discussed in Case 1, Section \ref{ch1subsec2}, we can first carry out the stabilization problem of finding $N_{\tau}$ so that the error time-delay system (\ref{c1.6}) is asymptotically stable. Note that the LMI of Lemma \ref{lem:Stabilization_interval_delay} is feasible for any delay $\tau \in [0.001,1.2]\cup[0.8,1.99]$. This result implies that we can compensate for the delayed measurements up to 1.99 seconds, and that we can asymptotically estimate the instantaneous state vector $z(t)=x(t)$ using delayed measurements $y(t)=x(t-\tau)$ for any $0.001 \leq \tau \leq 1.99\mathrm{s}$.

For illustrative purpose, let $\tau$ be given as $\tau=1\mathrm{s}$. The LMI of Lemma \ref{lem:stabilization-constant-delay} produces
\[N_{\tau}=\begin{pmatrix}  -0.5445 &-0.2188\\-0.2188 &-0.0850\end{pmatrix}.\]

With $N$ and $N_{\tau}$ as above, we can compute those stable eigenvalues \cite{wu} which are closest to the imaginary axis, which are at $ \{-0.4725\pm j0.2865\}$.

Substituting $N_{\tau}$ into (\ref{c1.20}), we obtain 
\[\bar{X}=\begin{pmatrix}
	\bar{G} &M
\end{pmatrix}=\begin{pmatrix}0.4359 &0.1745 &0.2181 &0.0869\\0.1745 &0.0694 &0.0869 &0.0357\end{pmatrix},\] and hence the rest of the observer parameters. Accordingly, the following observer is obtained to asymptotically estimate the instantaneous state vector $x(t)$
\begin{align}
	&\hat{z}(t)=\hat{x}(t)=w(t)+\begin{pmatrix} 0.2181 &0.0869\\0.0869 &0.0357\end{pmatrix}y(t),\nonumber\\
	&\dot{w}(t)=\begin{pmatrix}
		0.1 &1\\1 &-2 \end{pmatrix}w(t)+\begin{pmatrix}  -0.5445 &-0.2188\\-0.2188 &-0.0850\end{pmatrix}w(t-1)\nonumber\\&+\begin{pmatrix} 0.5445 &0.2188\\0.2188 &0.0850 \end{pmatrix}y(t)+\begin{pmatrix}-0.1378 &-0.0551\\-0.0551 &-0.0220\end{pmatrix}y(t-1)\nonumber\\
	&+\begin{pmatrix}
		1\\2\end{pmatrix}u(t)+\begin{pmatrix}
		-0.3918\\-0.1582\end{pmatrix}u(t-1).\nonumber
\end{align}
Figure \ref{fig1ch2} shows trajectories of $x_1(t)$, $\hat{x}_1(t)$ and $y_1(t)=x_1(t-1)$. It is clear
that asymptotic estimation has been achieved.
\begin{figure}[!h]
	\centering
	\includegraphics[width=0.9\linewidth]{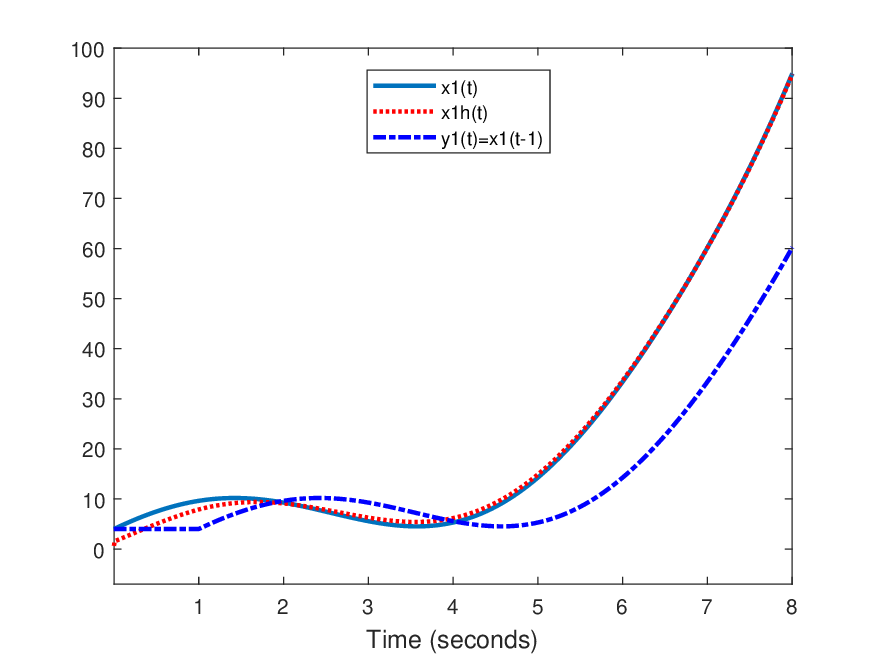}
	\caption{Trajectories of $x_1(t)$, $\hat{x}_1(t)$ and $y_1(t)=x_1(t-1)$}
	\label{fig1ch2}
\end{figure}
Figure \ref{fig2ch2} shows trajectories of $x_2(t)$, $\hat{x}_2(t)$ and $y_2(t)=x_2(t-1)$. It is clear
that asymptotic estimation has been achieved.
\begin{figure}[!h]
	\centering
	\includegraphics[width=0.9\linewidth]{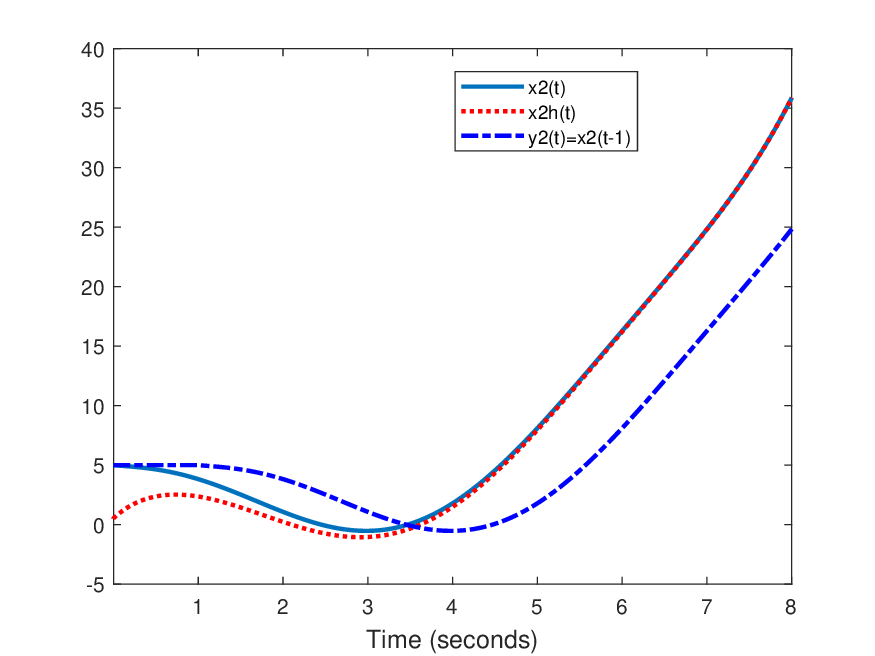}
	\caption{Trajectories of $x_2(t)$, $\hat{x}_2(t)$ and $y_2(t)=x_2(t-1)$}
	\label{fig2ch2}
\end{figure}

\textit{Numerical Example 2:} Let us relook at Numerical Example 1 but now we consider a more stringent situation where the output measurement contains only one delayed state variable $x_1(t-\tau)$, i.e.,
\[y(t)= \begin{pmatrix} 1 &0\end{pmatrix}x(t-\tau).\]

Again, our objective is to design an observer to estimate instantaneously the entire state vector $x(t)$ using the above delayed output measurement. With $C_{\tau}=\begin{pmatrix} 1 &0\end{pmatrix}$, we have
\[\begin{pmatrix} C_{\tau}\\C_{\tau}A \end{pmatrix}=\begin{pmatrix} 1 &0\\0.1 &1 \end{pmatrix}\]
which is a full-rank matrix. Hence, as discussed in Case 1, Section \ref{ch1subsec2}, we can first find $N_{\tau}$ so that the error time-delay system (\ref{c1.6}) is asymptotically stable. This was done in Example 1. Let $\tau$ be given as  $\tau=1\mathrm{s}$, and thus with $N_{\tau}$ as given in Example 1, we obtain the following observer 
\begin{align}
	\hat{z}(t)&=\hat{x}(t)=w(t)+\begin{pmatrix}0.2188\\0.0850\end{pmatrix}y(t),\nonumber\\
	\dot{w}(t)&=\begin{pmatrix}
		0.1 &1\\1 &-2 \end{pmatrix}w(t)+\begin{pmatrix}  -0.5445 &-0.2188\\-0.2188 &-0.0850\end{pmatrix}w(t-1)\nonumber\\&+\begin{pmatrix}0.6295\\0.2593 \end{pmatrix}y(t)+\begin{pmatrix}-0.1378\\-0.0551 \end{pmatrix}y(t-1)+\begin{pmatrix}
		1\\2\end{pmatrix}u(t)\nonumber\\&+\begin{pmatrix}-0.2188\\-0.0850\end{pmatrix}u(t-1).\nonumber
\end{align}

\textit{Numerical Example 3:} Let $A$, $B$ and $C_{\tau}$ be given as follows
\[A=\begin{pmatrix} 0.1 &1\\0 &0.5 \end{pmatrix}, \ B=\begin{pmatrix} 1\\2 \end{pmatrix}, \ C_{\tau}=\begin{pmatrix} 1 &0\end{pmatrix}.\]
Matrix $A$ is unstable with $\sigma(A)=\{0.1, 0.5\}$. In this example, the output measurement is delayed by a time delay, $\tau$, where
\[y(t)=\begin{pmatrix} 1 &0\end{pmatrix}x(t-\tau)=x_1(t-\tau).\]
Our objective is to design an observer to estimate instantaneously state variable $x_2(t)$ using the delayed output measurement $x_1(t-\tau)$. In this case, we have $F=\begin{pmatrix} 0 &1\end{pmatrix}$.

With the given $F$, condition (\ref{c1.9}) is satisfied, and we have
\[N=FAF^{-}=0.5\]
which is not Hurwitz. Therefore, we need to employ an observer with an internal delay, $N_{\tau}w(t-\tau)$, $N_{\tau}\neq 0$, to provide asymptotic stability of the error time-delay system (\ref{c1.6}).

Since $\begin{pmatrix} C_{\tau}\\C_{\tau}A \end{pmatrix}=\begin{pmatrix} 1 &0\\0.1 &1 \end{pmatrix}$ is a full-rank matrix, we can therefore first find $N_{\tau}$ so that the following error time-delay system 
\[\dot{e}(t)=0.5e(t)+N_{\tau}e(t-\tau)\]
is asymptotically stable. The above is a scalar time-delay system, and according to \cite{mori}, the exact stability conditions for this
system are
\[
0.5+N_{\tau}<0, \qquad N_{\tau}\geq -\frac{1}{\tau}.
\]
These conditions imply that the delay $\tau$ must satisfy
\[
\tau<2\,\mathrm{s}.
\]
Therefore, we can asymptotically estimate the instantaneous state $x_2(t)$ using delayed measurements $x_1(t-\tau)$ for any $0<\tau < 2\mathrm{s}$. For illustrative purpose, we pick $\tau=1\mathrm{s}$ and $N_{\tau}=-0.7$, and hence $\dot{e}(t)=0.5e(t)-0.7e(t-1)$ is asymptotically stable. Furthermore, based on \cite{wu}, its dominant eigenvalues are at $\{-0.4041\pm j0.5311\}$.

Now, with $N_{\tau}=-0.7$ and $F=\begin{pmatrix} 0 &1\end{pmatrix}$, from (\ref{c1.20}), we obtain
\[\begin{pmatrix}\bar{G}  &M\end{pmatrix}=\begin{pmatrix}-0.07 &0.7\end{pmatrix}.\]
Accordingly, the following first-order observer is obtained
\begin{align}
	\hat{z}(t)&=\hat{x}_2(t)=w(t)+0.7y(t),\nonumber\\
	\dot{w}(t)&=0.5w(t)-0.7w(t-1)+0.28y(t)-0.49y(t-1)\nonumber\\&+2u(t)-0.7u(t-1).\nonumber
\end{align}

\textit{Numerical Example 4:} This example is given to illustrate Case 2 as discussed in Section \ref{ch1subseccase2} where we can employ Lemma \ref{lem:output-stabilization-constant-delay} to solve a stabilization problem of finding $\bar{Z}\in \mathbb{R}^{m\times v}$ to stabilize the error time-delay system (\ref{c1.21}).

Let $A$, $B$ and $C_{\tau}$ be given as follows
\[A=\begin{pmatrix} 0.1 &1 &1\\0 &0.2 &1\\0 &-1 &-0.1 \end{pmatrix}, B=\begin{pmatrix} 1\\2\\3 \end{pmatrix}, \ C_{\tau}=\begin{pmatrix} 1 &0 &0\end{pmatrix}.\]

Here, again note that $A$ is an unstable matrix. In this example, there is only one output measurement, and it is delayed by a time delay, $\tau$, i.e.,
$$y(t)=\begin{pmatrix} 1 &0 &0\end{pmatrix}x(t)=x_1(t-\tau).$$
Our objective is to design an observer to estimate instantaneously state variables $x_2(t)$ and $x_3(t)$ using the delayed output measurement $y(t)=x_1(t-\tau)$. In this case, we have $z(t)=Fx(t)$, where $F=\begin{pmatrix} 0 &1 &0\\0 &0 &1\end{pmatrix}$.

In this example, with $F$ as given, condition (\ref{c1.9}) is satisfied and we have
\[N=FAF^{-}=\begin{pmatrix} 0.2 &1\\-1 &-0.1\end{pmatrix}.\]
Note that $N$ is an unstable matrix. On the other hand, due to having only one delayed output measurement, we now have  $\mathrm{rank}\begin{pmatrix}C_{\tau}\\ C_{\tau}A\end{pmatrix}=2$, which is $\mathrm{rank}\begin{pmatrix}C_{\tau}\\ C_{\tau}A\end{pmatrix}<n$. This situation is as discussed in Case 2, Section \ref{ch1subseccase2} where we need to find $Z$ to stabilize the error time-delay system (\ref{c1.16}).  

From (\ref{c1.15a})-(\ref{c1.15b}), we obtain $N_{\tau_2}$ as follows
\[N_{\tau_2}=\begin{pmatrix}0.3322 &0.3322\\0.3322 &0.3322\\0.0332 &0.0332\\-0.3322 &-0.3322\end{pmatrix},\] 
and $Z\in\mathbb{R}^{2\times 4}$ is arbitrary, to be determined such that the error time-delay system (\ref{c1.16}) is asymptotically stable.

It is clear from the above that $\mathrm{rank}(N_{\tau_2})=1$. Thus, by letting $Z=\begin{pmatrix}\bar{Z} &\mathbf 0_{2\times 3}\end{pmatrix}$, $\bar{Z}\in \mathbb{R}^{2\times 1}$, and $\bar{N}=\begin{pmatrix} 0.3322 &0.3322\end{pmatrix}$, we can consider the stabilization problem of finding $\bar{Z}$ so that the error time-delay system (\ref{c1.21}) is asymptotically stable. To solve this problem, for a given $\tau$, we use LMI of Lemma \ref{lem:output-stabilization-constant-delay}. On the other hand, with $N$ and $\bar{N}$ as given above, we found that LMI of Lemma \ref{lem:output-Stabilization_interval_delay} is feasible for any delay $\tau \in [0.001,1.2]\cup[0.8,2.25]$.

For illustrative purpose, let $\tau$ be given as $\tau=1\mathrm{s}$. We then use LMI of Lemma \ref{lem:output-stabilization-constant-delay} to obtain
\[\bar{Z}=\begin{pmatrix} -1.1384\\0.2522\end{pmatrix}, \ N_{\tau}=\bar{Z}\bar{N}=\begin{pmatrix}-0.3782 &-0.3782\\0.0838 &0.0838\end{pmatrix}.\]

Finally, we obtain the following second-order observer to estimate instantaneously both state variables $x_2(t)$ and $x_3(t)$
\begin{align}
	&\hat{z}(t)=\begin{pmatrix}\hat{x}_2(t)\\\hat{x}_3(t)\end{pmatrix}=w(t)+\begin{pmatrix} 0.3782\\-0.0838\end{pmatrix}y(t),\nonumber\\
	&\dot{w}(t)=\begin{pmatrix} 0.2 &1\\-1 &-0.1\end{pmatrix}w(t)+\begin{pmatrix}-0.3782 &-0.3782\\0.0838 &0.0838\end{pmatrix}w(t-1)\nonumber\\&+\begin{pmatrix}-0.0460\\-0.3615 \end{pmatrix}y(t)+\begin{pmatrix}-0.1114\\0.0247\end{pmatrix}y(t-1)+\begin{pmatrix}
		2\\3\end{pmatrix}u(t)\nonumber\\&+\begin{pmatrix}-0.3782\\0.0838\end{pmatrix}u(t-1).\nonumber
\end{align}

\section{Enlargement of $\tau$ using Delayed Output Measurement Vector $y(t)=C_{\tau}x(t-\tau)+C_hx(t-h)$}
\label{ch1sec5}
In Numerical Example 3, we showed that by using observer (\ref{c1.4a})-(\ref{c1.4b}), we can asymptotically estimate the instantaneous state $x_2(t)$ using delayed measurements $y(t)=x_1(t-\tau)$ for any $0<\tau < 2\mathrm{s}$. However, when $\tau\geq 2\mathrm{s}$, the estimation-error system $\dot{e}(t)=0.5e(t)+N_{\tau}e(t-\tau)$
cannot be guaranteed to be asymptotically stable using observer (\ref{c1.4a})-(\ref{c1.4b}). Note that increasing the time delay $\tau$ is of great practical significance, as it facilitates a broader range of applications. 

In this section, inspired by the recent result in \cite{trinhnam1}, we propose a method for enlarging $\tau$ by introducing the following delayed output measurement vector
\begin{align}
	\label{c1.22}
	y_e(t)=y(t-\alpha)=C_{\tau}x(t-h),
\end{align} 
where $\alpha=h-\tau$, $h>\tau$. Thus, $y_e(t)$ is obtained by delaying $y(t)$ by $\alpha$. We then consider the following new delayed measurement vector
\begin{align}
	\label{c1.23}
	y_a(t)=\begin{pmatrix}
		y(t)\\y_e(t)
	\end{pmatrix}=\bar{C}_{\tau}x(t-\tau)+\bar{C}_hx(t-h),\end{align} 
where
\[\bar{C}_{\tau}:=\begin{pmatrix}
	C_{\tau}\\\mathbf 0\end{pmatrix}, \ \bar{C}_h:=\begin{pmatrix}
	\mathbf 0\\ C_{\tau}\end{pmatrix}.\]
Before we jump into increasing $\tau$ by designing an observer around our new delayed vector $y_a(t)$, it is worth pausing for a key observation. This point, detailed below, is the primary motivation behind incorporating $y_a(t)$ into our design.

Let us consider the following time-delay systems
\begin{align}
	\label{c1.24a}\dot{e}(t)&=0.5e(t)-0.5e(t-2),\\
	\label{c1.24b} \dot{e}(t)&=0.5e(t)-0.8566e(t-2.3)+0.3509e(t-3).\end{align} 

\textit{Stability Analysis Comparison}: Based on the exact stability conditions defined by Mori \cite{mori}, the error equation (\ref{c1.24a}) is unstable at $\tau=2\text{s}$. In contrast, applying Lemma \ref{lem:output-Stabilization-two-delay} demonstrates that the error equation (\ref{c1.24b}) remains asymptotically stable at $\tau=2.3\text{s}$ and $h=3\text{s}$.

The primary takeaway is that with the same $0.5e(t)$ term in both equations, (\ref{c1.24b}) maintains stability under a larger delay $\tau$. By strategically leveraging both the $N_{\tau}e(t-\tau)$ and $N_he(t-h)$ terms (where $h > \tau$), we can achieve asymptotic stability for equation (\ref{c1.24b}) even as the delay $\tau$ increases.

This key observation motivates the development of a more general observer structure. By utilizing the new delayed measurement vector $y_a(t)$ and incorporating additional internal delay channels, this approach significantly enhances design flexibility.

In the following, let us design a general observer to estimate the functional $z(t)=Fx(t)$ using the following delayed measurement vector
\begin{align}
	\label{c1.25}
	&y(t)=C_{\tau}x(t-\tau)+C_hx(t-h),
\end{align}
where $C_{\tau},C_h\in \mathbb{R}^{p\times n}$ need not be of full row rank, and $h>\tau>0$. 

Note that (\ref{c1.25}) covers more practical scenarios and include (\ref{c1.23}) as a special case. In practice, delayed output measurements may involve multiple time delays. For example, let $n=2$ and $y(t)= \begin{pmatrix} x_1(t-\tau)\\x_2(t-h)\end{pmatrix}$. Here, there are two delayed output measurements: One is from state variable $x_1(t)$ being delayed by $\tau$; and the other measurement is from state variable $x_2(t)$ being delayed by $h$. In such case, we can express the output measurement vector $y(t)$ in the form (\ref{c1.25}), where
\[C_{\tau}=\begin{pmatrix} 1 &0\\0 &0\end{pmatrix}, \ C_h=\begin{pmatrix}0 &0\\0 &1\end{pmatrix}.\]

On the other hand, as for (\ref{c1.23}), let  $n=2$ and $y(t)= \begin{pmatrix} x_1(t-\tau)\\x_1(t-h)\end{pmatrix}$. Then we can express the output measurement vector $y(t)$ in the form (\ref{c1.25}), where
\[C_{\tau}=\begin{pmatrix} 1 &0\\0 &0\end{pmatrix}, \ C_h=\begin{pmatrix}0 &0\\1 &0\end{pmatrix}.\]

To estimate the functional (\ref{c1.3}) for linear system (\ref{c1.1}) using the delayed output vector (\ref{c1.25}), we consider the following observer with multiple delayed terms
\begin{align}
	\hat{z}(t)&=w(t)+My(t),\label{c1.26a}\\
	\dot{w}(t)&=Nw(t)+N_{\tau}w(t-\tau)+N_{h}w(t-h)+Gy(t)\nonumber\\&+G_{\tau}y(t-\tau)+G_hy(t-h)+Ju(t)\nonumber\\&+J_{\tau}u(t-\tau)+J_hu(t-h),\label{c1.26b}
\end{align}
where $w(\theta)=\rho(\theta),\ \forall \theta \in[-h,0]$,  $w(t) \in \mathbb{R}^m$, and $\hat{z}(t)$ is the estimate of $z(t)$. Matrices $M$, $N$, $N_{\tau}$, $N_h$, $G$, $G_{\tau}$, $G_h$, $J$, $J_{\tau}$ and $J_h$ are to be determined such that $\hat{z}(t)$ converges asymptotically to $z(t)$. 
\subsection{Observer Existence Conditions and Derivation of Observer Parameters}
\label{ch1subsec3}
Defining the estimation error vector $e(t)=\hat z(t)-z(t)$, the error dynamics are given by
\begin{align}
	\label{c1.27}
	\dot{e}(t)&=\dot{w}(t)+M\dot{y}(t)-F\dot{x}(t)\nonumber\\ \quad &=Ne(t)+N_{\tau}e(t-\tau)+N_he(t-h)+\mathcal{\bar C}_{1}u(t)\nonumber\\&+\mathcal{\bar C}_{2}u(t-\tau)+\mathcal{\bar C}_{3}u(t-h)+\mathcal{\bar C}_{4}x(t)+\mathcal{\bar C}_{5}x(t-\tau)\nonumber\\&+\mathcal{\bar C}_{6}x(t-h)+\mathcal{\bar C}_{7}x(t-2\tau)+\mathcal{\bar C}_{8}x(t-\tau-h)\nonumber\\&+\mathcal{\bar C}_{9}x(t-2h),
\end{align}
where\\ 
$\mathcal{\bar C}_1 =J-FB,\quad \mathcal{\bar C}_2 = J_{\tau}+MC_{\tau}B, \quad \mathcal{\bar C}_3=J_h+MC_hB, \quad \mathcal{\bar C}_4 =NF-FA$, \quad
$\mathcal{\bar C}_5 = N_{\tau}F+\bar{G}C_{\tau}+MC_{\tau}A$, \quad $\mathcal{\bar C}_6 = N_hF+\bar{G}C_h+MC_hA$, \quad $\mathcal{\bar C}_7 = \bar{G}_{\tau}C_{\tau},$ \quad $\mathcal{\bar C}_8 = \bar{G}_hC_{\tau}+\bar{G}_{\tau}C_h,$ \quad $\mathcal{\bar C}_9 = \bar{G}_hC_h,$ \quad $\bar{G}:=G-NM$, \quad $\bar{G}_{\tau}:=G_{\tau}-N_{\tau}M$, \quad $\bar{G}_h:=G_h-N_hM$,\\and\\
$\mathcal {\bar C}
:=
\begin{pmatrix}
	\mathcal {\bar C}_1 &
	\mathcal {\bar C}_2 &
	\mathcal {\bar C}_3 &
	\mathcal {\bar C}_4 &
	\mathcal {\bar C}_5 &\mathcal {\bar C}_6 &\mathcal {\bar C}_7 &\mathcal {\bar C}_8 & \mathcal {\bar C}_9
\end{pmatrix}.
$

The following theorem provides conditions for the existence of observer (\ref{c1.26a})-(\ref{c1.26b}).

\textit{\textbf{Theorem 2:}}
	\textit{Observer (\ref{c1.26a})-(\ref{c1.26b}) provides asymptotic estimation of the functional $z(t)=Fx(t)$ and yields
	estimation error dynamics that are decoupled from the plant state $x(\cdot)$ and the input $u(\cdot)$
	if $\mathcal {\bar C}=\bf 0$ and the following delay-dependent error dynamics
	\begin{align}
		\label{c1.28}&\dot{e}(t)=Ne(t)+N_{\tau}e(t-\tau)+N_he(t-h)
	\end{align}
	is asymptotically stable. In this case, the estimation error satisfies
	\[
	e(t)=\hat z(t)-z(t)\to {\bf 0} \quad \text{as} \quad t\to\infty
	\]
	for all admissible initial conditions and inputs $u(\cdot)$}.

\textit{\textbf{Proof:}} If $\mathcal {\bar C}=\bf 0$, then \eqref{c1.27} reduces to \eqref{c1.28}, so the error dynamics are
	decoupled from $x(\cdot)$ and $u(\cdot)$. If, in addition, \eqref{c1.28} is asymptotically stable, then $e(t)\to \bf 0$ as $t\to\infty$ for all admissible initial conditions and inputs $u(\cdot)$. This completes the proof.

We next derive observer parameters to satisfy conditions given in Theorem 2.

With regard to $\mathcal {\bar C}=\bf 0$, i.e., \[\begin{pmatrix}
	\mathcal {\bar C}_1 &
	\mathcal {\bar C}_2 &
	\mathcal {\bar C}_3 &
	\mathcal {\bar C}_4 &
	\mathcal {\bar C}_5 &\mathcal {\bar C}_6 &\mathcal {\bar C}_7 &\mathcal {\bar C}_8 & \mathcal {\bar C}_9
\end{pmatrix}=\bf 0.\]

Firstly, $\mathcal {\bar C}_1=\mathcal {\bar C}_2=\mathcal {\bar C}_3=\bf 0$ is satisfied by the following direct
choices of $J$, $J_{\tau}$ and $J_h$, respectively,
\begin{align}
	\label{c1.29}
	&J = FB, \quad J_{\tau}=-MC_{\tau}B, \quad J_h=-MC_hB.
\end{align}

Secondly, let us consider the requirement $\mathcal {\bar C}_7=\mathcal {\bar C}_8=\mathcal {\bar C}_9=\bf 0$. It can be expressed as follows
\begin{align}
	\label{c1.30}
	\begin{pmatrix}\bar{G}_{\tau} &\bar{G}_h\end{pmatrix} \begin{pmatrix}C_{\tau} &C_h &\bf 0\\\bf 0 &C_{\tau} &C_h\end{pmatrix}=\bf 0.
\end{align}
Based on the relationship in equation (\ref{c1.30}), we can simplify the requirements by setting both $\bar{G}_{\tau} = \mathbf{0}$ and $\bar{G}_h = \mathbf{0}$. This choice acts as a sufficient condition to satisfy (\ref{c1.30}), which in turn allows us to directly derive the values for $G_{\tau}$ and $G_h$ as follows
\begin{align}
	\label{c1.31}
	&G_{\tau}=N_{\tau}M, \quad G_h=N_hM.
\end{align}

As for $\mathcal {\bar C}_4=\bf 0$, as discussed in Section \ref{ch1sec3}, subject to the satisfaction of condition (\ref{c1.9}), we obtain $N=FAF^-$.

About the requirements $\mathcal {\bar C}_5=\mathcal {\bar C}_6=\bf 0$, let us first consider the case where $\mathrm{rank}\begin{pmatrix}C_{\tau} &C_h\\ C_{\tau}A &C_hA\end{pmatrix}=2n$.
Let us now express $\mathcal {\bar C}_5=\mathcal {\bar C}_6=\bf 0$ as follows
\begin{align}
	\label{c1.32}
	\bar{X}\bar{\Theta}_e= \Upsilon,
\end{align}
where
\[\bar{X}:=\begin{pmatrix}\bar{G}  &M\end{pmatrix},  \bar{\Theta}_e:= \begin{pmatrix}C_{\tau} &C_h\\C_{\tau}A &C_hA\end{pmatrix}, \Upsilon=-\begin{pmatrix}N_{\tau}F   &N_hF\end{pmatrix}.\]
Since $\mathrm{rank}\begin{pmatrix}C_{\tau} &C_h\\ C_{\tau}A &C_hA\end{pmatrix}=2n$, by Lemma 1, a solution
$\bar{X}$ always exists for any $N_{\tau}\neq \bf 0$ and $N_h\neq \bf 0$, and is given by
\begin{align}
	\label{c1.33}
	\bar{X}=\Upsilon\bar{\Theta}_e^-.
\end{align}

As in Case 1, Section \ref{ch1subsec2}, the above development shows that subject to $\mathrm{rank}\begin{pmatrix}C_{\tau} &C_h\\ C_{\tau}A &C_hA\end{pmatrix}=2n$, we can independently carry out the stabilization problem of finding $N_{\tau}$ and $N_h$ such that the error  time-delay system (\ref{c1.28}) with $N=FAF^-$ is asymptotically stable. 

As demonstrated in our stability analysis comparison between both time-delay systems (\ref{c1.24a})-(\ref{c1.24b}), we can use both the $N_{\tau}e(t-\tau)$ and $N_he(t-h)$ terms (where $h > \tau$) to achieve asymptotic stability for equation (\ref{c1.28}) with a larger $\tau$ than the case $N_h=\mathbf 0$. A sufficient condition for this stability is the LMI in Lemma \ref{lem:output-Stabilization-two-delay} feasible.  

Once $N_{\tau}$ and $N_h$ are found, from (\ref{c1.33}), we obtain $\bar{X}$. The remaining observer matrices $\bar{G}$ and $M$ are then recovered from the
corresponding blocks of $\bar{X}$. The rest of the observer gains are obtained as follows
\[
G=\bar G+NM, \ G_{\tau}=N_{\tau}M, \ G_h=N_hM, \ J = FB, \]

$J_{\tau}=-MC_{\tau}B$, $J_h=-MC_hB.$

\textit{\textbf{Remark 2:}} Consider $y_a(t)$ as defined by (\ref{c1.23}), we now show that $\mathrm{rank}\begin{pmatrix} C_{\tau}\\ C_{\tau}A\end{pmatrix}=n$ implies $\mathrm{rank}\begin{pmatrix}C_{\tau} &C_h\\ C_{\tau}A &C_hA\end{pmatrix}=2n$. With $\bar{C}_{\tau}:=\begin{pmatrix}
	C_{\tau}\\\mathbf 0\end{pmatrix}$ and $\bar{C}_h:=\begin{pmatrix}
	\mathbf 0\\ C_{\tau}\end{pmatrix}$, we have
\begin{align}
\mathrm{rank}\begin{pmatrix}\bar{C}_{\tau} &\bar{C}_h\\ \bar{C}_{\tau}A &\bar{C}_hA\end{pmatrix}&=\mathrm{rank}\begin{pmatrix}C_{\tau} &\mathbf 0\\ \mathbf 0 &C_{\tau}\\ C_{\tau}A &\mathbf 0\\\mathbf 0 &C_{\tau}A\end{pmatrix}\nonumber\\&=\mathrm{rank}\begin{pmatrix}C_{\tau} &\mathbf 0\\ C_{\tau}A &\mathbf 0\\\mathbf 0 &C_{\tau}\\ \mathbf 0 &C_{\tau}A\end{pmatrix}=2n. \nonumber
\end{align}

Thus, we can enlarge $\tau$ by utilizing the expanded delayed measurement vector $y_a(t)$ and using observer (\ref{c1.26a})-(\ref{c1.26b}) where $y(t)$ is replaced with $y_a(t)$.

\textit{\textbf{Remark 3:}} Similar to Case 2 as discussed in Section \ref{ch1subseccase2}, we can also consider the case where  $\mathrm{rank}\begin{pmatrix}C_{\tau} &C_h\\ C_{\tau}A &C_hA\end{pmatrix}<2n$. The detailed analysis is left to the reader to further investigate.

\textit{\textbf{Remark 4:}} Assume that there exist $l$ rows of $C_{\tau}$ such that 
$\mathrm{rank}\begin{pmatrix}C_l \\ C_lA\end{pmatrix}=n$, where $C_l$ is obtained by stacking those $l$ rows of $C_{\tau}$. Let us also rearrange and express $y(t)$  as follows
\[y(t)=\begin{pmatrix}y_l(t)\\y_r(t)\end{pmatrix}=\begin{pmatrix}C_lx(t)\\C_rx(t)\end{pmatrix},\]
where $y_l(t)\in \mathbb{R}^l$ and $y_r(t)\in \mathbb{R}^{p-l}$.

Subject to $\mathrm{rank}\begin{pmatrix}C_l \\ C_lA\end{pmatrix}=n$, and as discussed in Remark 2, to design an observer (\ref{c1.26a})-(\ref{c1.26b}) with the aim of enlarging $\tau$, it is sufficient for us to use the following reduced set of delayed measurement vector
\[y_a(t)=\begin{pmatrix}
	y_l(t)\\y_l(t-\alpha)
\end{pmatrix}=\bar{C}_{\tau}x(t-\tau)+\bar{C}_hx(t-h),\]
where $\alpha=h-\tau$, $h>\tau$, $\bar{C}_{\tau}:=\begin{pmatrix}
	C_l\\\mathbf 0\end{pmatrix}$ and $\bar{C}_h:=\begin{pmatrix}
	\mathbf 0\\ C_l\end{pmatrix}$. 

\subsection{Numerical Examples}
\textit{Numerical Example 5:} Let us now reconsider Numerical Example 3. We showed that by using observer (\ref{c1.4a})-(\ref{c1.4b}), we can asymptotically estimate the instantaneous state variable $x_2(t)$ using delayed measurements $y(t)=x_1(t-\tau)$ for any $0<\tau < 2\mathrm{s}$. However, when $\tau\geq 2\mathrm{s}$, the estimation-error system $\dot{e}(t)=0.5e(t)+N_{\tau}e(t-\tau)$
cannot be guaranteed to be asymptotically stable using observer (\ref{c1.4a})-(\ref{c1.4b}).

Now, let us introduce the following delayed output measurement
\[y(t-\alpha)=x_1(t-h),\]
where $\alpha=h-\tau$. Thus, $y(t-\alpha)$ is obtained by delaying $y(t)$ by $\alpha$. We then consider the following new delayed measurement vector
\[
y_a(t)=\begin{pmatrix}
	y(t)\\y(t-\alpha)
\end{pmatrix}=\bar{C}_{\tau}x(t-\tau)+\bar{C}_hx(t-h),\]
where
\[\bar{C}_{\tau}=\begin{pmatrix}
	1 &0\\0 &0\end{pmatrix}, \quad \bar{C}_h=\begin{pmatrix}
	0 &0\\1 &0\end{pmatrix}.\]

Now consider using the above expanded output vector $y_a(t)$ with observer   (\ref{c1.26a})-(\ref{c1.26b}) to estimate
$z(t)=\begin{pmatrix} 0 &1\end{pmatrix}x(t)=x_2(t)$. 

With $\bar{C}_{\tau}$ and $\bar{C}_h$ as given above, we can verify that $\mathrm{rank}\begin{pmatrix}\bar{C}_{\tau} &\bar{C}_h\\ \bar{C}_{\tau}A &\bar{C}_hA\end{pmatrix}=4$. Moreover, the LMI in Lemma \ref{lem:output-Stabilization-two-delay}
is feasible for the delays $\tau=2.3\mathrm{s}$
and $h=3\mathrm{s}$. The resulting observer parameters are
\[
N=0.5, \quad N_{\tau}=-0.8566, \quad N_h=0.3509,
\]
and the corresponding estimation-error time-delay system
\[
\dot{e}(t)=0.5e(t)-0.8566e(t-2.3)+0.3509e(t-3)
\]
is asymptotically stable.

From (\ref{c1.33}), we obtain
\[ \bar{X}=\begin{pmatrix}\bar{G}  &M\end{pmatrix}=\begin{pmatrix}-0.0857 &0.0351 &0.8566 &-0.3509 \end{pmatrix}.\]
Accordingly, the following first-order observer is obtained
\begin{align}
	&\hat{z}(t)=\hat{x}_2(t)=w(t)+\begin{pmatrix}0.8566 &-0.3509 \end{pmatrix}y_a(t),\nonumber\\
	&\dot{w}(t)=0.5w(t)-0.8566w(t-2.3)+0.3509w(t-3)\nonumber\\
	&+\begin{pmatrix}0.3427 &-0.1403 \end{pmatrix}y_a(t)+\begin{pmatrix}-0.7338 &0.3006 \end{pmatrix}y_a(t-2.3)\nonumber\\
	&+\begin{pmatrix}0.3006 &-0.1231 \end{pmatrix}y_a(t-3)+2u(t)-0.8566u(t-2.3)\nonumber\\
	&+0.3509u(t-3).\nonumber
\end{align}
Consequently, by employing the expanded output vector $y_a(t)$ in conjunction with the observer defined in (\ref{c1.26a})-(\ref{c1.26b}), we can asymptotically estimate the instantaneous functional$$z(t)=\begin{pmatrix} 0 & 1 \end{pmatrix}x(t)=x_2(t).$$This approach accommodates an output measurement delay of $\tau=2.3\text{s}$, surpassing the $2\text{s}$ limit achieved by the observer in (\ref{c1.4a})-(\ref{c1.4b}). This improvement serves as a clear demonstration of the innovation inherent in our proposed design.

\textit{Example 6:} This example illustrates the estimation of the desired functional
$z(t)$ using delayed measurements of the form (\ref{c1.25}), i.e.,
\[
y(t)=C_{\tau}x(t-\tau)+C_hx(t-h).
\]
Let $A$ and $B$ be given as follows
\[A=\begin{pmatrix} 0.5 &1\\-2 &2 \end{pmatrix}, \ B=\begin{pmatrix} 1\\1 \end{pmatrix}.\]
Let the output measurement vector be given as follows
\[y(t)=\begin{pmatrix} y_1(t)\\y_2(t)\end{pmatrix}=\begin{pmatrix} x_1(t-0.65)\\x_2(t-1.65)\end{pmatrix},\]
which takes the form
\[ y(t)=C_{\tau}x(t-0.65)+C_hx(t-1.65),\]where $C_{\tau}=\begin{pmatrix} 1 &0\\0 &0\end{pmatrix}$ and $C_h=\begin{pmatrix} 0 &0\\0 &1\end{pmatrix}$. Thus, there are two output measurements, $y_1(t)$ is from state variable $x_1(t)$ being delayed by $\tau=0.65\text{s}$; and the other measurement, $y_2(t)$, is from state variable $x_2(t)$ being delayed by $h=1.65\text{s}$. 

Our objective is to design an observer to estimate instantaneously the state vector $x(t)$, i.e., $z(t)=Fx(t)=x(t)$ where $F=I_2$. Here, with $F=I_2$, condition (\ref{c1.9}) is satisfied and $N=A$, which is, however, unstable.

Let us first highlight the significance of having both delayed measurements, $y_1(t)$ and $y_2(t)$, available for the design of an observer to estimate $z(t)$ using observer (\ref{c1.26a})-(\ref{c1.26b}). If only one of the output measurements is used, let say $y_1(t)$, then by using observer (\ref{c1.4a})-(\ref{c1.4b}), it is found that the LMI in Lemma \ref{lem:stabilization-constant-delay} is not feasible for $\tau=0.65\text{s}$ (the same problem with $h=1.65s$ for the case where only output measurement $y_2(t)$ is used).

Let us now  design a second-order observer of the form (\ref{c1.26a})-(\ref{c1.26b}) to estimate $z(t)$ using both output measurements $y_1(t)$ and $y_2(t)$.
With $C_{\tau}$ and $C_h$ as given above, we can verify that $\mathrm{rank}\begin{pmatrix}C_{\tau} &C_h\\ C_{\tau}A &C_hA\end{pmatrix}=4$. 
Therefore, as discussed in Section \ref{ch1subsec3}, we can independently carry out the stabilization problem of finding $N_{\tau}$ and $N_h$ such that the error  time-delay system (\ref{c1.28}) with $N=A$ is asymptotically stable. 

The LMI in Lemma \ref{lem:output-Stabilization-two-delay} is feasible for the delays $\tau=0.65\mathrm{s}$
and $h=1.65\mathrm{s}$. The resulting observer parameters are

$N=A$, $N_{\tau}=\begin{pmatrix} -0.3621 &  -1.0222\\
	2.0430 &  -1.8944 \end{pmatrix}$,

$N_h=\begin{pmatrix}  -0.2399 &   0.1801\\	-0.3589  &  0.0293\end{pmatrix}$,
\\
and the corresponding estimation-error time-delay system
\begin{align}
	\dot{e}(t)&=\begin{pmatrix} 0.5 &1\\-2 &2 \end{pmatrix}e(t)+\begin{pmatrix} -0.3621 &  -1.0222\\
		2.0430 &  -1.8944\end{pmatrix}e(t-0.65)\nonumber\\&+\begin{pmatrix}  -0.2399  &  0.1801\\ 	-0.3589  &  0.0293\end{pmatrix}e(t-1.65)
\end{align}
is asymptotically stable.

From (\ref{c1.33}), we obtain
\[ \bar{X}=\begin{pmatrix}\bar{G}  &M\end{pmatrix}=\begin{pmatrix}-0.1490 &0.0598 &1.0222 &-0.1200\\-2.9902 &0.3296 &1.8944 &-0.1795 \end{pmatrix}.\]
Accordingly, the following second-order observer is obtained
\begin{align}
	\hat{z}(t)&=\hat{x}(t)=w(t)+\begin{pmatrix}1.0222 &-0.1200\\1.8944 &-0.1795 \end{pmatrix}y(t),\nonumber\\
	\dot{w}(t)&=\begin{pmatrix} 0.5 &1\\-2 &2 \end{pmatrix}w(t)+\begin{pmatrix} -0.3621 &  -1.0222\\
		2.0430 &  -1.8944\end{pmatrix}w(t-0.65)\nonumber\\&+\begin{pmatrix}  -0.2399  &  0.1801\\ 	-0.3589  &  0.0293\end{pmatrix}w(t-1.65)\nonumber\\
	&+\begin{pmatrix}2.2564 &-0.1797\\-1.2458 &0.2106
	\end{pmatrix}y(t)\nonumber\\
	&+\begin{pmatrix}-2.3065 &0.2269\\-1.5003 &0.0949 \end{pmatrix}y(t-0.65)\nonumber\\
	&+\begin{pmatrix}0.0960 &-0.0035\\-0.3114 &0.0378 \end{pmatrix}y(t-1.65)\nonumber\\
	&+\begin{pmatrix}1\\1 \end{pmatrix}u(t)+\begin{pmatrix}-1.0222\\-1.8944 \end{pmatrix}u(t-0.65)\nonumber\\
	&+\begin{pmatrix}0.1200\\0.1795\end{pmatrix}u(t-1.65).\nonumber
\end{align}

\section{Implementing Time-Delay Compensators in Observer-Based Control Systems}
\label{ch1sec6}
Consider the linear system (\ref{c1.1}) where the measured output vector (\ref{c1.2}) is subject to time delay $\tau$. It is well-known that stabilizing controllers for linear systems
are frequently formulated in a state-feedback form, where
the control input depends linearly on the system state, $u(t)=Fx(t)$,
or,
\[
u(t)=z(t), \qquad z(t)=Fx(t).
\]
However, in many practical applications the full system state is
rarely directly measurable.
Limitations in sensing, communication constraints, distributed
architectures, and measurement delays often prevent direct access to
$x(t)$. Consequently, the implementation of stabilizing state-feedback
controllers requires the estimation of the control signal $z(t)$. Functional observers provide a direct means of estimating the
desired functional $z(t)$ without reconstructing the entire
state vector.
This contrasts with state-observer-based schemes, where the full state vector
$x(t)$ must first be estimated before the functional $z(t)$ can be
computed.

When the output vector (\ref{c1.2}) is subject to time delay $\tau$, our designed functional observer, which acts as a time-delay compensator, can effectively estimate the instantaneous functional $z(t)$. Hence, it can now be used to implement the stabilizing state-feedback control law $u(t)$.

In the following, we discuss the case where we use an observer (\ref{c1.4a})-(\ref{c1.4b}) of order $q>m$ to realize the state-feedback control law $u(t)=z(t)=Fx(t)$. This case is more general than the case where $q=m$. In this case, the estimation of the control signal is realized from (\ref{c1.extract}), ie.,
\[\hat{z}(t)=K\hat{z}_{\mathrm{aug}}(t), \quad K=\begin{pmatrix}I_m &\mathbf{0}_{m\times (q-m)}\end{pmatrix}. \]
Firstly, note that our designed observer ensures 	\[
e_{\mathrm{aug}}(t)=\hat z_{\mathrm{aug}}(t)-z_{\mathrm{aug}}(t)\to {\bf 0} \quad \text{as} \quad t\to\infty.
\]
By substituting \[u(t)=\hat{z}(t)=Ke_{\mathrm{aug}}(t)+K\bar{F}x(t)=Ke_{\mathrm{aug}}(t)+Fx(t)\] into (\ref{c1.1}), we obtain
\begin{align}
	\label{c1.34}
	\dot{x}(t)=Ax(t)+Bu(t)=(A+BF)x(t)+BKe_{\mathrm{aug}}(t).
\end{align}
Note that $F$ is a stabilizing state-feedback gain, designed so that $A+BF$ is Hurwitz. On the other hand, the delay-dependent error dynamics (\ref{c1.6}) is asymptotically stable. 

Let us now augment (\ref{c1.34}) and (\ref{c1.6}), where $e(t)$ is replaced with $e_{\mathrm{aug}}(t)$, to form the following \textit{state-error} system 
\begin{align}
	\label{c1.35}
	\begin{pmatrix}\dot{x}(t)\\\dot{e}_{\mathrm{aug}}(t)\end{pmatrix}&=\begin{pmatrix}A+BF &BK\\ \bf 0 &N\end{pmatrix}\begin{pmatrix}x(t)\\e_{\mathrm{aug}}(t)\end{pmatrix}\nonumber\\
	&+\begin{pmatrix}\bf 0 &\bf 0\\ \bf 0 &N_{\tau}\end{pmatrix} \begin{pmatrix}x(t-\tau)\\e_{\mathrm{aug}}(t-\tau)\end{pmatrix}.
\end{align}
Since $A+BF$ is Hurwitz and (\ref{c1.6}) is asymptotically stable, the state-error system (\ref{c1.35}) is asymptotically stable. Moreover, the characteristic roots of (\ref{c1.35}) are the union of the eigenvalues of $A+BF$ and the characteristic roots of (\ref{c1.6}).

On the other hand, by substituting \[u(t)=\hat{z}(t)=K\hat{z}_{\mathrm{aug}}(t)=Kw(t)+KMy(t)\] into (\ref{c1.1}) and (\ref{c1.4b}), we obtain the following closed-loop \textit{state-observer} system

\begin{align}
	\label{c1.36}
	\begin{pmatrix}\dot{x}(t)\\\dot{w}(t)\end{pmatrix}&=\begin{pmatrix}A &BK\\ \bf 0  &\tilde{N}\end{pmatrix}\begin{pmatrix}x(t)\\w(t)\end{pmatrix}\nonumber\\
	&+\begin{pmatrix}BKMC_{\tau} &\bf 0\\ \tilde{G}C_{\tau} &\tilde{N}_{\tau}\end{pmatrix} \begin{pmatrix}x(t-\tau)\\w(t-\tau)\end{pmatrix}\nonumber\\&+\begin{pmatrix}\bf 0 &\bf 0\\ \tilde{G}_{\tau}C_{\tau} &\bf 0\end{pmatrix} \begin{pmatrix}x(t-2\tau)\\w(t-2\tau)\end{pmatrix}, 
\end{align}
where

$\tilde{N}:=N+JK$,  $\tilde{N}_{\tau}:=N_{\tau}+J_{\tau}K$,  $\tilde{G}:=G+JKM$,

$\tilde{G}_{\tau}:=G_{\tau}+J_{\tau}KM.$ 

Figure (\ref{figc1.4}) shows a schematic diagram of the implementation of the closed-loop observer-based feedback control using delayed output vector $y(t)=C_{\tau}x(t-\tau)$.
\begin{figure}[!h]
	\centering
	\includegraphics[width=\linewidth]{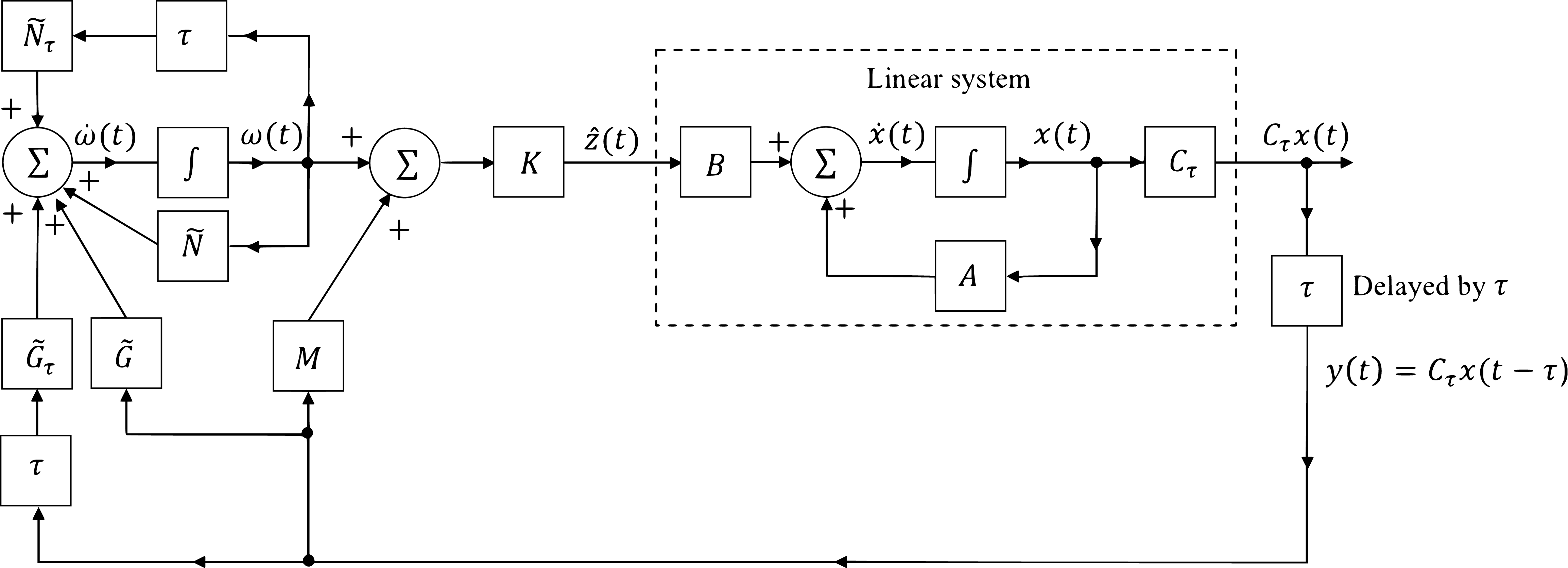}
	\caption{Schematic diagram of an observer-based feedback control scheme}
	\label{figc1.4}
\end{figure}

\textit{Numerical Example 7:} In this example, we consider linear system (\ref{c1.1}) with delayed output vector (\ref{c1.2}). We first design a state-feedback controller, $u(t)=z(t)=Fx(t)$, to stabilize the system. Then, we design an observer (\ref{c1.4a})-(\ref{c1.4b}) to estimate the control signal $\hat{z}(t)$. Finally, we use (\ref{c1.36}) to realize the closed-loop observer-based control system.

Let $A$, $B$ and $C_{\tau}$ be given as follows
\[A=\begin{pmatrix} 0 &1\\0.1 &0.2 \end{pmatrix}, \quad B=\begin{pmatrix}0\\1 \end{pmatrix}, \quad C_{\tau}=\begin{pmatrix}1 &0 \end{pmatrix}.\]

Matrix $A$ is unstable and thus let us design a stabilizing state-feedback controller, $u(t)=z(t)=Fx(t)$, to place the eigenvalues of $A+BF$ at $\{-0.5, -1\}$. This gives $F=\begin{pmatrix}-0.6 &-1.7 \end{pmatrix}$.

Let us now design an observer (\ref{c1.4a})-(\ref{c1.4b}) to estimate directly the functional $z(t)$. With $F$ and $A$ as above, however, condition (\ref{c1.9}) is not satisfied.
In this case, as discussed in Case 3, Section \ref{ch1subseccase3}, we increase the order of the designed observer by augmenting the orginal functional $z(t)$ with an extra functional, namely, $z_a(t)=Rx(t)$. Since the order of the observer equals the number of functionals being estimated, increasing the observer
order corresponds to enlarging the functional subspace.

For this example, let
\[
z_a(t)=\begin{pmatrix}0 &1 \end{pmatrix}x(t), \quad R=\begin{pmatrix}0 &1 \end{pmatrix}.
\]
The augmented functional is then defined as
\[
z_{\mathrm{aug}}(t)
= \begin{pmatrix}z(t)\\z_a(t) \end{pmatrix}=
\begin{pmatrix}
	F\\
	R
\end{pmatrix}x(t)=\bar{F}x(t),
\]
where $\bar F:=\begin{pmatrix}F\\ R\end{pmatrix}=\begin{pmatrix}-0.6 &-1.7\\0 &1\end{pmatrix}$ is full rank. Thus, condition (\ref{c1.rank}) is now satisfied and we also obtain
\[N=\bar{F}A\bar{F}^-=\begin{pmatrix}0.2833 &-0.4583\\-0.1667 &-0.0833\end{pmatrix}.\]

Note that $N$ is not Hurwitz. For this example, we have $\begin{pmatrix}C_{\tau}\\ C_{\tau}A\end{pmatrix}=I_2$. Therefore, as discussed in Case 1, Section \ref{ch1subsec2}, we can first carry out the stabilization problem of finding $N_{\tau}$ so that the error time-delay system (\ref{c1.6}) is asymptotically stable.  Note that the LMI of Lemma \ref{lem:Stabilization_interval_delay} is feasible for any delay $\tau \in [0.001,2.2]\cup[2.0,2.3]$.

For illustrative purpose, let $\tau$ be given as $\tau=1\mathrm{s}$. The LMI of Lemma \ref{lem:stabilization-constant-delay}  produces
\[N_{\tau}=\begin{pmatrix}-0.5079 &0.1035\\0.1217 &-0.1952\end{pmatrix}.\]

Finally, we obtain the following second-order observer to estimate the functional $z(t)$
\begin{align}
	\hat{z}(t)&=w_1(t)-0.9670y(t),\nonumber\\
	\dot{w}(t)&=\begin{pmatrix}0.2833 &-0.4583\\-0.1667 &-0.0833\end{pmatrix}w(t)\nonumber\\
	&+\begin{pmatrix}-0.5079 &0.1035\\0.1217 &-0.1952\end{pmatrix}w(t-1)\nonumber\\&+\begin{pmatrix}-0.7630\\0.2007 \end{pmatrix}y(t)+\begin{pmatrix}0.5327\\-0.1962\end{pmatrix}y(t-1)\nonumber\\
	&+\begin{pmatrix}
		-1.7\\1\end{pmatrix}u(t).\nonumber
\end{align}
The trajectories of $x_1(t)$ using functional observer-based control and state-feedback control schemes are shown in Figure \ref{fig3ch2}, along with the step-like reference signal $r(t)$.
It is important to note that observer-based control initially degrades the system's transient response due to non-zero reconstruction errors. However, this performance penalty decays exponentially as the estimation error $e(t)$ vanishes, ultimately validating the effectiveness of the functional observer.

\begin{figure}[!h]
	\centering
	\includegraphics[width=0.9\linewidth]{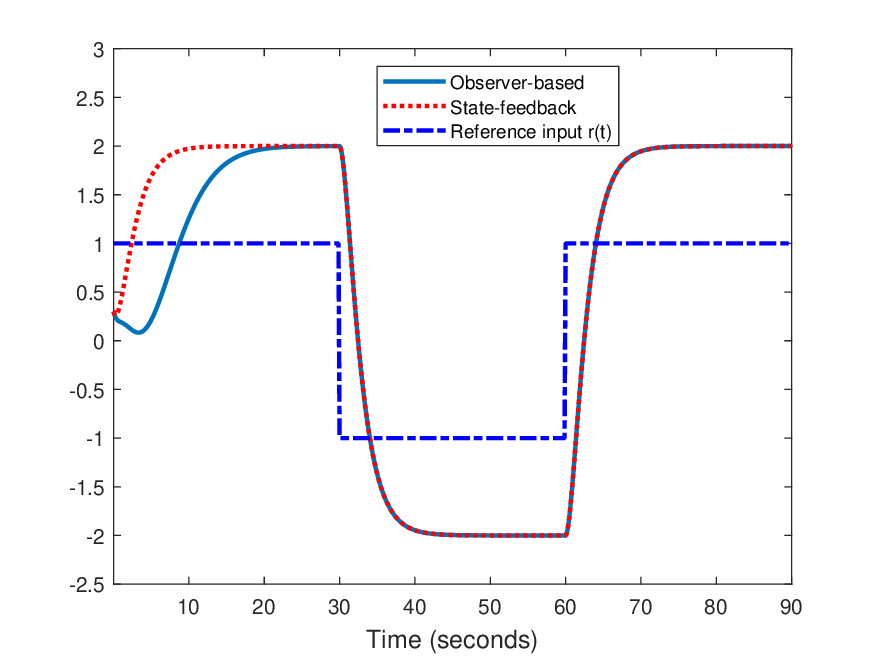}
	\caption{Trajectories of $x_1(t)$ using functional observer-based control and state-feedback control, and reference signal $r(t)$}
	\label{fig3ch2}
\end{figure}
\section{Conclusion}
\label{ch1sec7}
This paper has established a comprehensive framework for designing functional observers for linear systems subject to output measurement delays.  Moving beyond traditional methodologies, the proposed observer generates an estimate $\hat{z}(t)$ that predicts the current state functional $Fx(t)$ using delayed data. By neutralizing sensing latency, the observer serves as a potent time-delay compensator, effectively expanding the practical utility of functional observer theory.

The proposed observer architecture utilizes multiple delayed components to maintain estimation accuracy in the presence of latency. A pivotal contribution of this paper is the method for increasing the permissible time delay, $\tau$, in output measurements without sacrificing the asymptotic stability of the estimation-error system. Stability is achieved by integrating both $N_{\tau}e(t-\tau)$ and $N_he(t-h)$ terms ($h > \tau$) into the estimation-error dynamic equation. This configuration ensures asymptotic stability even as the primary delay $\tau$ increases.
These findings enable a more versatile observer structure and design. The design offers superior flexibility over standard models by employing a generalized observer structure with internal delay channels. By utilizing an augmented measurement vector,$$y_a(t)=\begin{pmatrix} y(t) \\ y(t-\alpha) \end{pmatrix}$$where $y(t-\alpha)$ is generated by further delaying the original vector, the approach significantly broadens the available design space.

\section{Appendix}
The following lemmas will be used in this paper.

\begin{lemma}\label{lemp1} \cite{rao}
	Let $\Theta\in\mathbb{R}^{k\times n}$ and $\Upsilon\in\mathbb{R}^{m\times n}$
	be given matrices, and consider the matrix equation
	\begin{equation}
		\label{p1}
		X\Theta = \Upsilon,
	\end{equation}
	where $X\in\mathbb{R}^{m\times k}$ is unknown.
	
	A solution exists if and only if
	\begin{equation}
		\label{p2}
		\mathrm{rank}\begin{pmatrix}\Theta\\\Upsilon \end{pmatrix}=\mathrm{rank}\begin{pmatrix}\Theta\end{pmatrix}.
	\end{equation}
	In this case, the general solution is given by
	\begin{equation}
		\label{p3}
		X = \Upsilon \Theta^{-}
		+ Z\big(I_k - \Theta \Theta^{-}\big),
	\end{equation}
	where
	$Z\in\mathbb{R}^{m\times k}$ is arbitrary.
\end{lemma}

\begin{lemma}\label{lemp2}
	Let $\Theta\in\mathbb{R}^{k\times n}$
	be a given matrix, and consider the matrix equation
	\begin{equation}
		\label{p4}
		X\Theta = \mathbf 0_{m\times n},
	\end{equation}
	where $X\in\mathbb{R}^{m\times k}$ is unknown.
	
	A non-trivial solution $X\neq \mathbf 0$ exists if and only if there exists an orthogonal basis of the left null-space of $\Theta$.  The general solution to (\ref{p4}) is given by
	\begin{equation}
		\label{p5}
		X=Z(I_k-\Theta\Theta^{-}),
	\end{equation}
	where
	$Z\in\mathbb{R}^{m\times k}$ is arbitrary.
\end{lemma}

The existence of an orthogonal basis of the left null-space of $\Theta$ is equivalent to $\Theta\Theta^{-}\neq I_k$.

\subsection{Stability criteria for time-delay systems}
\label{sec1.2.1}
\begin{lemma}[The  reciprocally convex combination inequality, \cite{Park:11}]\label{lem_Park} For given
	positive integers  $n,m$, a scalar $\alpha \in (0,1)$, a $n\times n$-matrix $R>0$,
	two $n\times m$-matrices $W_1, W_2$. Define, for all vector $\xi\in \mathbb{R}^m$, the
	function $\Theta(\alpha,R)$ given by:
	$$\Theta(\alpha,R)=\frac{1}{\alpha}\xi^{\mathsf T}W_1^{\mathsf T}RW_1\xi+\frac{1}{1-\alpha}\xi^{\mathsf T}W_2^{\mathsf T}RW_2\xi.$$
	If there is a matrix $S\in \mathbb{R}^{n\times n}$ such that
	$$M=\begin{pmatrix} R &S\\ \star &R\end{pmatrix}>0,$$
	then  the following inequality holds
	$$\min_{\alpha\in (0,1)}\Theta(\alpha,R) \geq
	\begin{pmatrix} W_1\xi\\ W_2\xi\end{pmatrix}^{\mathsf T}
	\begin{pmatrix} R &S\\ \star &R\end{pmatrix}
	\begin{pmatrix} W_1\xi\\ W_2\xi\end{pmatrix}.$$
\end{lemma}

\begin{lemma}[The Wirtinger-based integral inequality, \cite{Seuret:13}]\label{lem_Sueret}For a given $n\times n$-matrix $R>0$,
	any differentiable function $\varphi: [a,b] \rightarrow \mathbb{R}^n$, then  the following inequality holds
\begin{align}
	\int_{a}^b\dot{\varphi}^{\mathsf T}(u)R\dot{\varphi}(u)du & \geq \frac{1}{b-a}(\varphi(b)-\varphi(a))^{\mathsf T}R(\varphi(b)-\varphi(a))
	\nonumber\\
	&+\frac{12}{b-a}\Omega^{\mathsf T}R\Omega,\nonumber
\end{align}
	where $\Omega=\frac{\varphi(b)+\varphi(a)}{2}-\frac{1}{b-a}\int_a^b\varphi(u)du.$
\end{lemma}

In the following development, we present a comprehensive treatment of stability analysis and stabilization for systems involving multiple delays. By leveraging powerful tools--specifically the reciprocally convex combination inequality and the Wirtinger-based integral inequality, alongside free-weighting matrix and delay-partitioning techniques--we obtain enhanced stability and stabilizability results for complex time-delay configurations. Numerical examples will demonstrate that incorporating additional delayed feedback terms effectively increases the maximum admissible time delay. This key finding has served as the basis for the development of novel observer schemes in this paper, which offer relaxed existence conditions by expanding the permissible delay margin in the state evolution.

Consider the following linear system with a constant time delay
\begin{align}
	\label{p6a}
	&\dot{e}(t)=Ne(t)+N_{\tau}e(t-\tau),\\
	\label{p6b}
	&e(t)=\phi(t),\ \forall t \in [-\tau, 0],
\end{align}
where $e(t)\in \mathbb{R}^n$ is the state vector, $\phi(t)$ is the initial function, $N, N_{\tau}\in \mathbb{R}^{n\times n}$ are constant matrices.

The time delay $\tau>0$ is considered for the following two cases:
\begin{itemize}
	\item{\it Case 1}: $\tau>0$ is a known constant delay; and	
	\item{\it Case 2}: $\tau>0$ is unknown but belongs to a known interval, i.e,  $\tau \in [\underline{\tau}, \overline{\tau}]$, where $0<\underline{\tau}< \overline{\tau}$ are two given scalars.
\end{itemize}
In the following, we present some stability conditions, in the form of LMIs, for system (\ref{p6a}).

First, the following lemma presents a stability condition for system (\ref{p6a}) under Case 1. Let us denote
\begin{align}
	&v_i:=\big(\mathbf 0_{n\times (i-1)n}\ \ I_n\ \ \mathbf 0_{n\times (4-i)n}\big)^{\mathsf T},
	\ i=1,\dots,4,\
	\nonumber\\
	&\Pi_1:= \big(v_1\ \ \tau v_3\big),  \quad \Pi_2:= \big(v_4\ \ v_1-v_2\big),\ \nonumber\\
	&\Gamma:=\big(v_1-v_2\quad \sqrt{3}(v_1+v_2-2v_3)\big),
	\ \ \tilde{R}:=\mathrm{diag}(R,R),\nonumber\\
	&{\mathcal N}^{\mathsf T}
	:=
	\big(N\quad N_{\tau}\quad \mathbf 0_{n\times n}\quad -I_n\big)
	\in\mathbb R^{n\times 4n}, \nonumber\\
	&\zeta_0^{\mathsf T}(t):=
	\left(
	e^{\mathsf T}(t)\quad
	\int_{t-\tau}^t e^{\mathsf T}(s)ds
	\right),  \nonumber\\
	&\xi^{\mathsf T}(t)	:=
	\left(
	e^{\mathsf T}(t)\quad
	e^{\mathsf T}(t-\tau)\quad
	\frac{1}{\tau}\int_{t-\tau}^te^{\mathsf T}(s)ds\quad
	\dot e^{\mathsf T}(t)
	\right).\nonumber	
\end{align}	
\begin{lemma}\label{lem:Sta_constant_delay}		
	Suppose that there exist a $2n\times 2n$ matrix $P>0$,
	two $n\times n$ matrices $Q>0$, $R>0$,
	two $n\times n$ matrices $X,Y$,
	such that the following LMI holds:
	\begin{align}\label{eq:LMI}
		\Theta + \mathrm{sym}\big((v_1X+v_4Y)\mathcal{N}^{\mathsf T}\big)<0,					
	\end{align}
	where 	
	\begin{align}\label{eq:theta}
		\Theta:=&\mathrm{sym}\big(\Pi_1P\Pi_2^{\mathsf T}\big)
		+v_1Qv_1^{\mathsf T}
		-v_2Qv_2^{\mathsf T}
		+\tau v_4 Rv_4^{\mathsf T} -\frac{1}{\tau}\Gamma \tilde{R} \Gamma^{\mathsf T}.				 
	\end{align}
	Then, system \eqref{p6a} is asymptotically stable.	
\end{lemma}

\textit{Proof:}\ \ Consider the Lyapunov-Krasovskii functional
\begin{align}\label{eq:LKF}
	V(t,e_t)
	=&\,
	\zeta_0^{\mathsf T}(t)P\zeta_0(t)
	+\int_{t-\tau}^te^{\mathsf T}(s)Qe(s)ds
	\nonumber\\
	&+\int_{-\tau}^{0}\int_{\theta}^0
	\dot{e}^{\mathsf T}(t+u)R\dot{e}(t+u)dud\theta.
\end{align}
Since $P > 0$, $Q > 0$, $R > 0$ and $||e(t)||\leq ||\zeta_0(t)||$, we have $0<\sigma_{\min}(P)$ and
\begin{eqnarray}\label{eq:lambda_min}
	\sigma_{\min}(P)||e(t)||^2\leq \sigma_{\min}(P)||\zeta_0(t)||^2\leq V(t,e_t), \ \forall t\geq 0.
\end{eqnarray}
Taking the derivative of $V(.,.)$ along the trajectory of \eqref{p6a}, we  obtain		
\begin{align}
	\dot{V}(t,e_t) = {} & 2\zeta_0^{\mathsf T}(t)P\dot{\zeta}_0(t) + e^{\mathsf T}(t)Qe(t) \nonumber \\
	& - e^{\mathsf T}(t-\tau)Qe(t-\tau) + \tau\dot{e}^{\mathsf T}(t)R\dot{e}(t) \nonumber \\
	& - \int_{t-\tau}^{t}\dot{e}^{\mathsf T}(s)R\dot{e}(s)ds \nonumber \\
	= {} & \xi^{\mathsf T}(t) \Bigl( \mathrm{sym}\big(\Pi_1 P \Pi_2^{\mathsf T}\big) + v_1 Q v_1^{\mathsf T} \nonumber \\
	& - v_2 Q v_2^{\mathsf T} + \tau v_4 R v_4^{\mathsf T} \Bigr) \xi(t) \nonumber \\
	& - \int_{t-\tau}^{t}\dot{e}^{\mathsf T}(s)R\dot{e}(s)ds. \label{eq:derivative-V}
\end{align}
By Lemma \ref{lem_Sueret}, we have
\begin{align}\label{eq:Wirtinger_ineq}
	&-\int_{t-\tau}^{t}\dot{e}^{\mathsf T}(s)R\dot{e}(s)ds \leq -\xi^{\mathsf T}(t)\frac{1}{\tau}\Gamma \tilde{R} \Gamma^{\mathsf T}\xi(t).
\end{align}	
Furthermore, system \eqref{p6a} can be rewritten in the form:
\begin{align}\label{p6a-rewritten}
	0=Ne(t)+N_{\tau}e(t-\tau)-\dot{e}(t)={\mathcal N}^{\mathsf T}\xi(t).
\end{align}	
Multiplying \eqref{p6a-rewritten} on the left by $2\xi^{\mathsf T}(t)(v_1X+v_4Y)$, i.e., using the free-weighting matrix technique, yields
\begin{align}\label{eq:descriptor}
	0=2\xi^{\mathsf T}(t)(v_1X+v_4Y){\mathcal N}^{\mathsf T}\xi(t).
\end{align}	
From \eqref{eq:theta},  \eqref{eq:derivative-V}, \eqref{eq:Wirtinger_ineq} and \eqref{eq:descriptor} combining with \eqref{eq:LMI}, we arrive at
\begin{align}\label{derivative-V<0}
	\dot{V}(t,e_t)\leq \xi^{\mathsf T}(t)\Big(\Theta+ \mathrm{sym}\big((v_1X+v_4Y)\mathcal{N}^{\mathsf T}\Big)\xi(t) \leq 0.
\end{align}
From \eqref{eq:lambda_min}, \eqref{derivative-V<0}	and the Lyapunov-Krasovskii Theorem \cite{Fridman:14}, we conclude that system (\ref{p6a}) is asymptotically stable.  This completes the proof of Lemma \ref{lem:Sta_constant_delay}.	\\

The next lemma presents a stability condition for system \eqref{p6a} under Case 2. Let us denote
\begin{align*}
	&v_i := \big(\mathbf 0_{n\times (i-1)n}\ \ I_n\ \ \mathbf 0_{n\times (8-i)n}\big)^{\mathsf T},\quad i=1,\dots,8,\\
	&\Pi_1(\tau):= \big(v_1\ \ \underline{\tau}v_5\ \ (\tau-\underline{\tau})v_6 \ \ (\overline{\tau}-\tau)v_7\big), \\
	&\Pi_2 := \big(v_8\ \ v_1-v_2\ \ v_2-v_3\ \ v_3-v_4 \big), \\
	&\Gamma_1 := \big(v_1-v_2\quad \sqrt{3}(v_1+v_2-2v_5)\big),\\
	&\Gamma_2:= \big(v_2-v_3\quad \sqrt{3}(v_2+v_3-2v_6)\big),\\
	&\Gamma_3:= \big(v_3-v_4\quad \sqrt{3}(v_3+v_4-2v_7)\big),\\
	&\Gamma_{2,3}:= \big(\Gamma_2\quad \Gamma_3\big),\\
	&\mathcal N^{\mathsf T}:= \big(N\quad \mathbf 0_{n\times n}\quad  N_{\tau}\quad \mathbf 0_{n\times 4n}\quad -I_n\big) \in\mathbb{R}^{n\times 8n},\\
	&\zeta_0^{\mathsf T}(t):= \Bigl(e^{\mathsf T}(t)\quad  \int_{t-\underline{\tau}}^t e^{\mathsf T}(s)\,ds \\
	&\qquad \qquad \int_{t-\tau}^{t-\underline{\tau}} e^{\mathsf T}(s)\,ds \quad  \int_{t-\overline{\tau}}^{t-\tau} e^{\mathsf T}(s)\,ds \Bigr),\\
	&\xi^{\mathsf T}(t) := \Bigl( e^{\mathsf T}(t)\quad e^{\mathsf T}(t-\underline{\tau}) \quad e^{\mathsf T}(t-\tau) \quad e^{\mathsf T}(t-\overline{\tau}) \\
	&\qquad \qquad \frac{1}{\underline{\tau}}\int_{t-\underline{\tau}}^t e^{\mathsf T}(s)\,ds \quad \frac{1}{\tau-\underline{\tau}}\int_{t-\tau}^{t-\underline{\tau}} e^{\mathsf T}(s)\,ds \\
	&\qquad \qquad \frac{1}{\overline{\tau}-\tau}\int_{t-\overline{\tau}}^{t-\tau} e^{\mathsf T}(s)\,ds\quad \dot e^{\mathsf T}(t) \Bigr).
\end{align*}

\begin{lemma}\label{lem:Sta_interval_delay}
	Suppose that there exist a ${4n\times 4n}$ matrix \(P>0\), five ${n\times n}$ matrices $Q_1>0$, $Q_2>0$, $Q_3>0$, $R_1>0$, $R_2>0$, two ${n\times n}$ matrices \(X, Y\), and a ${2n\times 2n}$ matrix \(S\) such that the following LMIs hold:
	\begin{align}\label{eq:LMI1}
		&\Theta_1(\tau)+\mathrm{sym}\big((v_1X+v_8Y)\mathcal{N}^{\mathsf T}\big) < 0,\quad \tau \in \{\underline{\tau}, \overline{\tau}\}, \\
		&\Xi(\tilde{R}_2,S)>0, \label{eq:LMI_Xi1}
	\end{align}
	where

\begin{align}
	\tilde{R}_1&:=\mathrm{diag}(R_1,R_1),\nonumber\\
	\tilde{R}_2&:=\mathrm{diag}(R_2,R_2), \nonumber\\
	\Xi(\tilde{R}_2,S)&:=\begin{pmatrix}
		\tilde{R}_2 & S\\
		\star & \tilde{R}_2
	\end{pmatrix}, \nonumber\\
		\Theta_1(\tau)&:=
		 \mathrm{sym}\big(\Pi_1(\tau)P\Pi_2^{\mathsf T}\big) +v_1Q_1v_1^{\mathsf T}
		-v_2(Q_1-Q_2)v_2^{\mathsf T}\nonumber\\
		&
		-v_3(Q_2-Q_3)v_3^{\mathsf T}-v_4Q_3v_4^{\mathsf T}
		+\underline{\tau}\, v_8 R_1v_8^{\mathsf T}\nonumber\\
		&
		+(\overline{\tau}-\underline{\tau})v_8 R_2v_8^{\mathsf T}-\frac{1}{\underline{\tau}}\Gamma_1 \tilde{R}_1 \Gamma_1^{\mathsf T}\nonumber\\
		&
		-\frac{1}{\overline{\tau}-\underline{\tau}}\Gamma_{2,3}\, \Xi(\tilde{R}_2,S)\, \Gamma_{2,3}^{\mathsf T}.\label{Theta_1}
\end{align}
	Then, the system \eqref{p6a} is asymptotically stable for any delay \(\tau \in [\underline{\tau}, \overline{\tau}]\).
\end{lemma}

\textit{Proof:} Let us consider the Lyapunov-Krasovskii functional
\begin{align}\label{eq:LKF1}
	V(t,e_t)
	=&\,
	\zeta_0^{\mathsf T}(t)P\zeta_0(t)
	+\int_{t-\underline{\tau}}^te^{\mathsf T}(s)Q_1e(s)ds\nonumber\\
	&
	+\int_{t-\tau}^{t-\underline{\tau}}e^{\mathsf T}(s)Q_2e(s)ds+\int_{t-\overline{\tau}}^{t-\tau}e^{\mathsf T}(s)Q_3e(s)ds\nonumber\\
	&
	+\int_{-\underline{\tau}}^{0}\int_{\theta}^0
	\dot{e}^{\mathsf T}(t+u)R_1\dot{e}(t+u)dud\theta\nonumber\\
	&+\int_{-\overline{\tau}}^{-\underline{\tau}}\int_{\theta}^0
	\dot{e}^{\mathsf T}(t+u)R_2\dot{e}(t+u)dud\theta.
\end{align}
Since $P,\ Q_1,\ Q_2,\ Q_3,\  R_1, \ R_2>0$ and $||e(t)||\leq ||\zeta_0(t)||$, we have $0<\sigma_{\min}(P)$ and
\begin{eqnarray}\label{eq:lambda_min1}
	\sigma_{\min}(P)||e(t)||^2\leq \sigma_{\min}(P)||\zeta_0(t)||^2\leq V(t,e_t), \ \forall t\geq 0.
\end{eqnarray}	
Taking the derivative of $V(.,.)$ along the trajectory of \eqref{p6a}, we  obtain
\begin{align}\label{eq:derivative-V1}
	\dot{V}(t,e_t)	=\ & 2\zeta_0^{\mathsf T}(t)P\dot{\zeta}_0(t)+ e^{\mathsf T}(t)Q_1e(t)\nonumber\\
	&-e^{\mathsf T}(t-\underline{\tau})Q_1e(t-\underline{\tau})+ e^{\mathsf T}(t-\underline{\tau})Q_2e(t-\underline{\tau})\nonumber\\
	&-e^{\mathsf T}(t-\tau)Q_2e(t-\tau)+ e^{\mathsf T}(t-\tau)Q_3e(t-\tau)\nonumber\\
	&-e^{\mathsf T}(t-\overline{\tau})Q_3e(t-\overline{\tau})+\underline{\tau}\dot{e}^{\mathsf T}(t)R_1\dot{e}(t)\nonumber\\
	&-\int_{t-\underline{\tau}}^{t}\dot{e}^{\mathsf T}(s)R_1\dot{e}(s)ds
+(\overline{\tau}-\underline{\tau})\dot{e}^{\mathsf T}(t)R_2\dot{e}(t)\nonumber\\
&-\int_{t-\overline{\tau}}^{t-\underline{\tau}}\dot{e}^{\mathsf T}(s)R_2\dot{e}(s)ds
	\nonumber\\
	=\ &\xi^{\mathsf T}(t)\Big(\mathrm{sym}\big(\Pi_1(\tau)P\Pi_2^{\mathsf T}\big)
	+v_1Q_1v_1^{\mathsf T}\nonumber\\
	&
	-v_2(Q_1-Q_2)v_2^{\mathsf T}-v_3(Q_2-Q_3)v_3^{\mathsf T}-v_4Q_3v_4^{\mathsf T}\nonumber\\
	&+\underline{\tau} v_8 R_1v_8^{\mathsf T} +(\overline{\tau}-\underline{\tau})v_8 R_2v_8^{\mathsf T}\Big)\xi(t)\nonumber\\
	&-\int_{t-\underline{\tau}}^{t}\dot{e}^{\mathsf T}(s)R_1\dot{e}(s)ds-\int_{t-\overline{\tau}}^{t-\underline{\tau}}\dot{e}^{\mathsf T}(s)R_2\dot{e}(s)ds\nonumber\\
	=\ &\xi^{\mathsf T}(t)\Big(\mathrm{sym}\big(\Pi_1(\tau)P\Pi_2^{\mathsf T}\big)
	+v_1Q_1v_1^{\mathsf T}\nonumber\\
	&
	-v_2(Q_1-Q_2)v_2^{\mathsf T}-v_3(Q_2-Q_3)v_3^{\mathsf T}\nonumber\\
	&-v_4Q_3v_4^{\mathsf T}+\underline{\tau} v_8R_1v_8^{\mathsf T}
	+(\overline{\tau}-\underline{\tau})v_8 R_2v_8^{\mathsf T}\Big)\xi(t)\nonumber\\
	&-\int_{t-\underline{\tau}}^{t}\dot{e}^{\mathsf T}(s)R_1\dot{e}(s)ds-\int_{t-\tau}^{t-\underline{\tau}}\dot{e}^{\mathsf T}(s)R_2\dot{e}(s)ds\nonumber\\
	&
	-\int_{t-\overline{\tau}}^{t-\tau}\dot{e}^{\mathsf T}(s)R_2\dot{e}(s)ds.
\end{align}	
By Lemma \ref{lem_Sueret}, we have
\begin{align}\label{eq:Wirtinger_ineq1_1}
	&-\int_{t-\underline{\tau}}^{t}\dot{e}^{\mathsf T}(s)R_1\dot{e}(s)ds \leq -\xi^{\mathsf T}(t)\frac{1}{\underline{\tau}}\Gamma_1 \tilde{R}_1 \Gamma_1^{\mathsf T}\xi(t),\\
	&-\int_{t-\tau}^{t-\underline{\tau}}\dot{e}^{\mathsf T}(s)R_2\dot{e}(s)ds \leq -\xi^{\mathsf T}(t)\frac{1}{\tau-\underline{\tau}}\Gamma_2 \tilde{R}_2 \Gamma_2^{\mathsf T}\xi(t),\label{eq:Wirtinger_ineq1_2}\\
	&-\int_{t-\overline{\tau}}^{t-\tau}\dot{e}^{\mathsf T}(s)R_2\dot{e}(s)ds \leq -\xi^{\mathsf T}(t)\frac{1}{\overline{\tau}-\tau}\Gamma_3 \tilde{R}_2 \Gamma_3^{\mathsf T}\xi(t).\label{eq:Wirtinger_ineq1_3}
\end{align}
Using \eqref{eq:Wirtinger_ineq1_2}, \eqref{eq:Wirtinger_ineq1_3}, \eqref{eq:LMI_Xi1} and Lemma \ref{lem_Park}, we deduce that
\begin{align}\label{eq:Wirtinger_ineq1_4}
	&-\int_{t-\tau}^{t-\underline{\tau}}\dot{e}^{\mathsf T}(s)R_2\dot{e}(s)ds -\int_{t-\overline{\tau}}^{t-\tau}\dot{e}^{\mathsf T}(s)R_2\dot{e}(s)ds \nonumber\\
	\leq
	&-\xi^{\mathsf T}(t)\frac{1}{\overline{\tau}-\underline{\tau}}  \Gamma_{2,3}\ \Xi(\tilde{R}_2,S)\ \Gamma_{2,3}^{\mathsf T}\xi(t).
\end{align}
Furthermore, system
\eqref{p6a} can be rewritten in the form:
\begin{align}\label{p6a-rewritten1}
	0=Ne(t)+N_{\tau}e(t-\tau)-\dot{e}(t)={\mathcal N}^{\mathsf T}\xi(t).
\end{align}	
Multiplying \eqref{p6a-rewritten1} on the left by $2\xi^{\mathsf T}(t)(v_1X+v_8Y)$, i.e., using the free-weighting matrix technique, yields
\begin{align}\label{eq:descriptor1}
	0=2\xi^{\mathsf T}(t)(v_1X+v_8Y){\mathcal N}^{\mathsf T}\xi(t).
\end{align}	
Using \eqref{Theta_1},  \eqref{eq:derivative-V1}, \eqref{eq:Wirtinger_ineq1_1}, \eqref{eq:Wirtinger_ineq1_4}, and \eqref{eq:descriptor1}, we arrive at
\begin{align}\label{derivative-V1<0}
	\dot{V}(t,e_t)\leq \xi^{\mathsf T}(t)\Big(\Theta_1(\tau)+\mathrm{sym}\big((v_1X+v_8Y)\mathcal{N}^{\mathsf T}\big)\Big)\xi(t).
\end{align}
On the other hand, since $\Theta_1(\tau)+\mathrm{sym}\big((v_1X+v_8Y)\mathcal{N}^{\mathsf T}\big)$ is affine, and hence convex, with respect to $\tau$,
it follows from condition \eqref{eq:LMI1} that
\begin{align}\label{Theta_1<0}
	\Theta_1(\tau)+\mathrm{sym}\big((v_1X+v_8Y)\mathcal{N}^{\mathsf T}\big) <0, \quad \forall \tau \in [\underline{\tau}, \overline{\tau}].	
\end{align}
By \eqref{eq:lambda_min1}, \eqref{derivative-V1<0}, \eqref{Theta_1<0},	together with the Lyapunov--Krasovskii theorem \cite{Fridman:14}, we conclude that the system (\ref{p6a}) is asymptotically stable for any delay $\tau \in [\underline{\tau}, \overline{\tau}]$.  This completes the proof of Lemma \ref{lem:Sta_interval_delay}.\\

Since $\tau \in [\underline{\tau}, \overline{\tau}]$, we can use convex parameter-dependent matrices to further improve the stability criterion in Lemma \ref{lem:Sta_interval_delay}, i.e., we can obtain a stability condition that yields a larger allowable range of the time delay. The improved stability criterion is presented in the following lemma.
\begin{lemma}\label{lem:impro_sta_interval_delay}		
	Suppose that there exist two $4n\times 4n$ matrices of the form
	$$P_1=\begin{pmatrix}
		P_{11}^1 & P_{12}\\
		\star & P_{22}
	\end{pmatrix}>0, \quad P_2=\begin{pmatrix}
		P_{11}^2 & P_{12}\\
		\star & P_{22}
	\end{pmatrix}>0,$$
	where $P^1_{11}, P^2_{11}, P_{12},P_{22} \in \mathbb{R}^{2n\times 2n}$,
	ten $n\times n$ matrices $Q_{11}>0,$ $Q_{12}>0,$ $Q_{21}>0,$ $Q_{22}>0,$ $Q_{31}>0,$ $Q_{32}>0,$ $R_{11}> 0$, $R_{12}>0$, $R_{21}>0$, $R_{22}>0$,
	two $2n\times 2n$ matrices $S_1, S_2$, and four $n\times n$ matrices $X_1, X_2, Y_1,  Y_2$, such that the following LMIs hold:
	\begin{align}\label{eq:LMI2}
		&\Theta_2(\tau)+\mathrm{sym}\big((v_1X(\tau)+v_8Y(\tau))\mathcal{N}^{\mathsf T}\big) <0, \quad \tau \in \{\underline{\tau}, \overline{\tau}\},\\
		& \Xi(\tilde{R}_2(\tau),S(\tau))>0, \quad \tau \in \{\underline{\tau},\overline{\tau}\},\label{eq:LMI_Xi2}
	\end{align}
	where
	\begin{align}\label{Theta_2}
		\Theta_2(\tau):=
		& \mathrm{sym}\big(\Pi_1(\tau)P(\tau)\Pi_2^{\mathsf T}\big)
		+v_1Q_1(\tau)v_1^{\mathsf T}\nonumber\\&
		-v_2(Q_1(\tau)-Q_2(\tau))v_2^{\mathsf T}-v_3(Q_2(\tau)-Q_3(\tau))v_3^{\mathsf T}\nonumber\\&-v_4Q_3(\tau)v_4^{\mathsf T}
		+\underline{\tau} v_8 R_1(\tau)v_8^{\mathsf T}
		\nonumber\\
		& +(\overline{\tau}-\underline{\tau})v_8 R_2(\tau)v_8^{\mathsf T}
		-\frac{1}{\underline{\tau}}\Gamma_1 \tilde{R}_1(\tau) \Gamma_1^{\mathsf T}\nonumber\\&		 -\frac{1}{\overline{\tau}-\underline{\tau}}\Gamma_{2,3}\ \Xi(\tilde{R}_2(\tau),S(\tau))\ \Gamma_{2,3}^{\mathsf T},
	\end{align}
	\begin{align*}
			P(\tau)&:=(\tau-\underline{\tau})P_1 +(\overline{\tau}-\tau)P_2,\\
			Q_1(\tau)&:=(\tau-\underline{\tau})Q_{11} +(\overline{\tau}-\tau)Q_{12},\\
			Q_2(\tau)&:=(\tau-\underline{\tau})Q_{21} +(\overline{\tau}-\tau)Q_{22},\\
			Q_3(\tau)&:=(\tau-\underline{\tau})Q_{31} +(\overline{\tau}-\tau)Q_{32},\\
			R_1(\tau)&:=(\tau-\underline{\tau})R_{11} +(\overline{\tau}-\tau)R_{12},\\
			R_2(\tau)&:=(\tau-\underline{\tau})R_{21} +(\overline{\tau}-\tau)R_{22},\\
			\tilde{R}_1(\tau)&:=\mathrm{diag}(R_1(\tau),R_1(\tau)),\\
			\tilde{R}_2(\tau)&:=\mathrm{diag}(R_2(\tau),R_2(\tau)),\\
			S(\tau)&:=(\tau-\underline{\tau})S_{1} +(\overline{\tau}-\tau)S_{2},\\
			\Xi(\tilde{R}_2(\tau),S(\tau))&:=\begin{pmatrix}
				\tilde{R}_2(\tau) & S(\tau)\\
				\star &\tilde{R}_2(\tau)
			\end{pmatrix},\\
			X(\tau)&:=(\tau-\underline{\tau})X_1 +(\overline{\tau}-\tau)X_2,\\
			Y(\tau)&:=(\tau-\underline{\tau})Y_1 +(\overline{\tau}-\tau)Y_2.
	\end{align*}	
	\noindent Then, system \eqref{p6a} is asymptotically stable for any delay $\tau \in [\underline{\tau}, \overline{\tau}]$.	
\end{lemma}

\textit{Proof:}  Since $P(\tau)$ is affine, and consequently convex, with respect to $\tau$, it follows, from two conditions $P_1>0$, $P_2>0$, that
$$P(\tau)>0,\quad \forall \tau\in[\underline{\tau},\overline{\tau}].$$
Similarly, we also have for all $\tau\in[\underline{\tau},\overline{\tau}],$
\begin{align*}
	Q_1(\tau)>0,\ Q_2(\tau)>0,\ Q_3(\tau)>0,\\
	 R_1(\tau)>0, \ R_2 (\tau)>0,\ \Xi(\tilde{R}_2(\tau),S(\tau))>0.
\end{align*}
Noting that  $\Theta_2(\tau)+\mathrm{sym}\big((v_1X(\tau)+v_8Y(\tau))\mathcal{N}^{\mathsf T}\big)$ is affine, and hence convex, with respect to $\tau$, it follows from condition \eqref{eq:LMI2}   that $$\Theta_2(\tau)+\mathrm{sym}\big((v_1X(\tau)+v_8Y(\tau))\mathcal{N}^{\mathsf T}\big)< 0,\quad  \forall \tau\in[\underline{\tau},\overline{\tau}].$$
Following the same arguments as in the proof of Lemma \ref{lem:Sta_interval_delay}, we conclude that system (\ref{p6a}) is asymptotically stable for any delay $\tau \in [\underline{\tau}, \overline{\tau}]$. This completes the proof of Lemma \ref{lem:impro_sta_interval_delay}.\\

We now present a numerical example to illustrate the feasibility and effectiveness of the results established in Lemmas \ref{lem:Sta_constant_delay}, \ref{lem:Sta_interval_delay}, and \ref{lem:impro_sta_interval_delay}.

\textit{Example A1:}
Consider system \eqref{p6a} with $N=\begin{pmatrix}
	0 & 1\\
	-2 & 0.1
\end{pmatrix}$ and
$N_{\tau}=\begin{pmatrix}
	0 & 0\\
	1 &0
\end{pmatrix}$.

By Lemma \ref{lem:Sta_constant_delay}, the system \eqref{p6a} is asymptotically stable with a maximum allowable time delay of $\tau = 1.54$. By Lemmas \ref{lem:Sta_interval_delay} and \ref{lem:impro_sta_interval_delay}, the maximum allowable time delays are $\overline{\tau} = 1.66$ and $\overline{\tau} = 1.67$, corresponding to $\tau \in [0.8, 1.66]$ and $\tau \in [0.5, 1.67]$, respectively (see Table \ref{table:ranges}). It can be seen that, compared with Lemma \ref{lem:Sta_constant_delay}, both Lemmas \ref{lem:Sta_interval_delay} and \ref{lem:impro_sta_interval_delay} provide larger maximum allowable time delays. This demonstrates that the stability conditions presented in Lemmas \ref{lem:Sta_interval_delay} and \ref{lem:impro_sta_interval_delay} are less conservative than the condition presented in Lemma \ref{lem:Sta_constant_delay}.

{\large \begin{table}[!h]
		\caption{Ranges of admissible time delays}
		\label{table:ranges}
		\vspace{0.2cm}
		\centering{\begin{tabular}{|l|c|c|}
				\hline
				$\underline{\tau}$\ $\backslash$\ {\rm methods}\qquad       & \qquad {\rm Lemma} \ref{lem:Sta_interval_delay}  \qquad  \qquad             &\qquad  {\rm Lemma} \ref{lem:impro_sta_interval_delay}     \qquad  \qquad              \\
				\hline
				
				\hline
				$\underline{\tau}=0.2$                                    &  [0.2,\ 1.42]     &  [0.2,\ 1.62]       \\
				\hline
				$\underline{\tau}=0.3$                                    &  [0.3,\ 1.55]     &  [0.3,\ 1.65]       \\
				\hline
				$\underline{\tau}=0.5$                                    &  [0.5,\ 1.65]     &  [0.5,\ 1.67]       \\
				\hline
				$\underline{\tau}=0.8$                                    &  [0.8,\ 1.66]     &  [0.8,\ 1.66]       \\
				\hline
				$\underline{\tau}=1.2$                                    &  [1.2,\ 1.64]     &  [1.2,\ 1.64]       \\
				\hline	 		
			\end{tabular}
		}	 	
\end{table}}

Moreover, from Table \ref{table:ranges}, it can be observed that, compared with Lemma \ref{lem:Sta_interval_delay}, Lemma \ref{lem:impro_sta_interval_delay} provides a wider range of admissible time delays, particularly, when $\underline{\tau}$ is small. This indicates that the stability condition in Lemma \ref{lem:impro_sta_interval_delay} is less conservative than that in Lemma \ref{lem:Sta_interval_delay}. However, the computational complexity of the stability condition in Lemma \ref{lem:impro_sta_interval_delay} is significantly higher than that in Lemma \ref{lem:Sta_interval_delay}.

Similar to the above, we also derive stability conditions for systems with multiple time delays. For simplicity, we consider a system with three time delays as follows:
\begin{align}
	\label{p7a}
	&\dot{e}(t)=Ne(t)+N_{1}e(t-\tau_1)+N_{2}e(t-\tau_2)+N_{3}e(t-\tau_3),\\
	\label{p7b}
	&e(t)=\phi(t), \quad \forall t \in [-\tau_3, 0],
\end{align}
where \(e(t)\in \mathbb{R}^n\) is the state vector, and \(\phi(t)\) is the initial function. The scalars \(0<\tau_1<\tau_2<\tau_3\) represent time delays, and \(N, N_{1}, N_2, N_3 \in \mathbb{R}^{n\times n}\) are constant matrices.

Let us denote
\begin{align*}
	v_i &:= \big(\mathbf 0_{n\times (i-1)n}\ \ I_n\ \ \mathbf 0_{n\times (8-i)n}\big)^{\mathsf T},\quad i=1,\dots,8,\\
	\Pi_1 &:= \big(v_1\ \ \tau_1v_5\ \ (\tau_2-\tau_1)v_6 \ \ (\tau_3-\tau_2)v_7\big),\\
	\Pi_2 &:= \big(v_8\ \ v_1-v_2\ \ v_2-v_3\ \ v_3-v_4 \big),\\
	\Gamma_1 &:= \big(v_1-v_2\quad \sqrt{3}(v_1+v_2-2v_5)\big),\\
	\Gamma_2 &:= \big(v_2-v_3\quad \sqrt{3}(v_2+v_3-2v_6)\big),\\
	\Gamma_3 &:= \big(v_3-v_4\quad \sqrt{3}(v_3+v_4-2v_7)\big),\\
	\mathcal{N}^{\mathsf T} &:= \big(N\quad N_1\quad N_2\quad N_3\quad \mathbf 0_{n\times 3n}\quad -I_n\big) \in\mathbb{R}^{n\times 8n},\\
	\zeta_0^{\mathsf T}(t) &:= \Bigl(e^{\mathsf T}(t)\quad  \int_{t-\tau_1}^t e^{\mathsf T}(s)\,ds \\
	&\quad  \int_{t-\tau_2}^{t-\tau_1} e^{\mathsf T}(s)\,ds \quad  \int_{t-\tau_3}^{t-\tau_2} e^{\mathsf T}(s)\,ds \Bigr),\\
	\xi^{\mathsf T}(t) &:= \Bigl( e^{\mathsf T}(t)\quad e^{\mathsf T}(t-\tau_1)\quad e^{\mathsf T}(t-\tau_2)\quad e^{\mathsf T}(t-\tau_3) \\
	&\quad \frac{1}{\tau_1}\int_{t-\tau_1}^t e^{\mathsf T}(s)\,ds \quad \frac{1}{\tau_2-\tau_1}\int_{t-\tau_2}^{t-\tau_1} e^{\mathsf T}(s)\,ds \\
	&\quad \frac{1}{\tau_3-\tau_2}\int_{t-\tau_3}^{t-\tau_2} e^{\mathsf T}(s)\,ds\quad \dot e^{\mathsf T}(t) \Bigr).
\end{align*}

The following lemma presents an LMI-based stability condition for system \eqref{p7a}.
\begin{lemma}\label{lem:Sta_multi}
	Suppose there exist a $4n\times 4n$-matrix $P>0$, six  $n\times n$-matrices
	$Q_1> 0$, $Q_2> 0$, $Q_3> 0$, $R_1> 0$, $R_2> 0$, $R_3> 0$,
	two $n\times n $-matrices $X, Y$, such that the following LMI holds:
	\begin{align}\label{eq:LMI3}			 \Theta_3+\mathrm{sym}\big((v_1X+v_8Y)\mathcal{N}^{\mathsf T}\big)<0,
	\end{align}
	where $
	\tilde{R}_1:=\mathrm{diag}(R_1,R_1),\quad
	\tilde{R}_2:=\mathrm{diag}(R_2,R_2),$ $
	\tilde{R}_3:=\mathrm{diag}(R_3,R_3),$ and	
	\begin{align}
		\Theta_3 := {} & \mathrm{sym}\big(\Pi_1 P \Pi_2^{\mathsf T}\big) + v_1 Q_1 v_1^{\mathsf T} - v_2(Q_1 - Q_2)v_2^{\mathsf T} \nonumber \\
		& - v_3(Q_2 - Q_3)v_3^{\mathsf T} - v_4 Q_3 v_4^{\mathsf T} + \tau_1 v_8 R_1 v_8^{\mathsf T} \nonumber \\
		& + (\tau_2 - \tau_1) v_8 R_2 v_8^{\mathsf T} + (\tau_3 - \tau_2) v_8 R_3 v_8^{\mathsf T} \nonumber \\
		& - \frac{1}{\tau_1}\Gamma_1 \tilde{R}_1 \Gamma_1^{\mathsf T} - \frac{1}{\tau_2 - \tau_1}\Gamma_2 \tilde{R}_2 \Gamma_2^{\mathsf T} \nonumber \\
		& - \frac{1}{\tau_3 - \tau_2}\Gamma_3 \tilde{R}_3 \Gamma_3^{\mathsf T}. \label{Theta3}
	\end{align}
	Then, system (\ref{p7a}) is asymptotically stable.
\end{lemma}
\textit{Proof:} Consider the Lyapunov--Krasovskii functional
\begin{align}\label{eq:LKF3}
	V(t,e_t) = &\ \zeta_0^{\mathsf T}(t)P\zeta_0(t)
	+\int_{t-\tau_1}^te^{\mathsf T}(s)Q_1e(s)ds
	\nonumber \\
	&+\int_{t-\tau_2}^{t-\tau_1}e^{\mathsf T}(s)Q_2e(s)ds+\int_{t-\tau_3}^{t-\tau_2}e^{\mathsf T}(s)Q_3e(s)ds\nonumber \\
	&+\int_{-\tau_1}^{0}\int_{\theta}^0\dot{e}^{\mathsf T}(t+u)R_1\dot{e}(t+u)dud\theta\nonumber \\ &+\int_{-\tau_2}^{-\tau_1}\int_{\theta}^0\dot{e}^{\mathsf T}(t+u)R_2\dot{e}(t+u)du d\theta \nonumber \\
	&+\int_{-\tau_3}^{-\tau_2}\int_{\theta}^0\dot{e}^{\mathsf T}(t+u)R_3\dot{e}(t+u)du d\theta.
\end{align}
Since $P, Q_1, Q_2, Q_3, R_1, R_2, R_3>0$ and $||e(t)||\leq ||\zeta_0(t)||$, we have $0<\sigma_{\min}(P)$ and
\begin{eqnarray}\label{eq:lambda_min3}
	\sigma_{\min}(P)||e(t)||^2\leq \sigma_{\min}(P)||\zeta_0(t)||^2\leq V(t,e_t), \ \forall t\geq 0.
\end{eqnarray}	
Taking the derivative of $V$ along the trajectory of \eqref{p7a}, we  obtain
\begin{align}
	\dot{V}(t,e_t) = {} & 2\zeta_0^{\mathsf T}(t)P\dot{\zeta}_0(t) + e^{\mathsf T}(t)Q_1e(t) \nonumber \\
	& - e^{\mathsf T}(t-\tau_1)(Q_1-Q_2)e(t-\tau_1) \nonumber \\
	& - e^{\mathsf T}(t-\tau_2)(Q_2-Q_3)e(t-\tau_2) \nonumber \\
	& - e^{\mathsf T}(t-\tau_3)Q_3e(t-\tau_3) + \tau_1\dot{e}^{\mathsf T}(t)R_1\dot{e}(t) \nonumber \\
	& - \int_{t-\tau_1}^{t}\dot{e}^{\mathsf T}(s)R_1\dot{e}(s)ds + (\tau_2-\tau_1)\dot{e}^{\mathsf T}(t)R_2\dot{e}(t) \nonumber \\
	& - \int_{t-\tau_2}^{t-\tau_1}\dot{e}^{\mathsf T}(s)R_2\dot{e}(s)ds + (\tau_3-\tau_2)\dot{e}^{\mathsf T}(t)R_3\dot{e}(t) \nonumber \\
	& - \int_{t-\tau_3}^{t-\tau_2}\dot{e}^{\mathsf T}(s)R_3\dot{e}(s)ds \nonumber \\
	= {} & \xi^{\mathsf T}(t) \Big( \mathrm{sym}\big(\Pi_1 P \Pi_2^{\mathsf T}\big) + v_1 Q_1 v_1^{\mathsf T} \nonumber \\
	& - v_2(Q_1-Q_2)v_2^{\mathsf T} - v_3(Q_2-Q_3)v_3^{\mathsf T} - v_4 Q_3 v_4^{\mathsf T} \nonumber \\
	& + \tau_1 v_8 R_1 v_8^{\mathsf T} + (\tau_2-\tau_1)v_8 R_2 v_8^{\mathsf T} \nonumber \\
	& + (\tau_3-\tau_2)v_8 R_3 v_8^{\mathsf T} \Big) \xi(t) \nonumber \\
	& - \int_{t-\tau_1}^{t}\dot{e}^{\mathsf T}(s)R_1\dot{e}(s)ds - \int_{t-\tau_2}^{t-\tau_1}\dot{e}^{\mathsf T}(s)R_2\dot{e}(s)ds \nonumber \\
	& - \int_{t-\tau_3}^{t-\tau_2}\dot{e}^{\mathsf T}(s)R_3\dot{e}(s)ds. \label{eq:derivative-V3}
\end{align}
By Lemma \ref{lem_Sueret}, we have
\begin{align}\label{eq:Wirtinger_ineq3_1}
	&-\int_{t-\tau_1}^{t}\dot{e}^{\mathsf T}(s)R_1\dot{e}(s)ds \leq -\xi^{\mathsf T}(t)\frac{1}{\tau_1}\Gamma_1 \tilde{R}_1 \Gamma_1^{\mathsf T}\xi(t),\\
	&-\int_{t-\tau_2}^{t-\tau_1}\dot{e}^{\mathsf T}(s)R_2\dot{e}(s)ds \leq -\xi^{\mathsf T}(t)\frac{1}{\tau_2-\tau_1}\Gamma_2 \tilde{R}_2 \Gamma_2^{\mathsf T}\xi(t),\label{eq:Wirtinger_ineq3_2}\\
	&-\int_{t-\tau_3}^{t-\tau_2}\dot{e}^{\mathsf T}(s)R_3\dot{e}(s)ds \leq -\xi^{\mathsf T}(t)\frac{1}{\tau_3-\tau_2}\Gamma_3 \tilde{R}_3 \Gamma_3^{\mathsf T}\xi(t).\label{eq:Wirtinger_ineq3_3}
\end{align}
By the free-weighting matrix technique, we also have
\begin{align}\label{eq:descriptor3}
	0=2\xi^{\mathsf T}(t)(v_1X+v_8Y){\mathcal N}^{\mathsf T}\xi(t).
\end{align}	
From \eqref{Theta3},  \eqref{eq:derivative-V3}-\eqref {eq:descriptor3}, and \eqref{eq:LMI3}, it follows that
\begin{align}\label{derivative-V3<0}
	\dot{V}(t,e_t)\leq \xi^{\mathsf T}(t)\Theta_3 \xi(t) \leq 0.
\end{align}
From \eqref{eq:lambda_min3}, \eqref{derivative-V3<0},	 and the Lyapunov--Krasovskii Theorem \cite{Fridman:14}, we conclude that the system \eqref{p7a} is asymptotically stable.  This completes the proof of Lemma \ref{lem:Sta_multi}.

\textit{Remark A1:}
We now reconsider system \eqref{p6a} in Case 1. By setting
\(\tau_1=\frac{\tau}{3}\), \(\tau_2=\frac{2\tau}{3}\), \(\tau_3=\tau\),
\(N_1=N_2=\mathbf{0}_{n\times n}\), and \(N_3=N_{\tau}\), Lemma \ref{lem:Sta_multi} can also be applied to derive a stability condition for system \eqref{p6a}. Since more information on the state vector is utilized, this derived stability condition is less conservative than that in Lemma \ref{lem:Sta_constant_delay}.		

\textit{Example A2:} Let us reconsider system \eqref{p6a} with matrices
\[
N=\begin{pmatrix}
	0 & 1\\
	-2 & 0.1
\end{pmatrix}, \quad
N_{\tau}=\begin{pmatrix}
	0 & 0\\
	1 & 0
\end{pmatrix}.
\]
By Lemma \ref{lem:Sta_constant_delay}, system \eqref{p6a} is asymptotically stable for a maximum allowable time delay \(\tau = 1.54\). By the proposed method, system \eqref{p6a} remains stable for a larger maximum allowable time delay of \(\tau = 1.69\). This indicates that the time-delay partitioning technique (Remark A1) helps to significantly reduce the conservatism of the derived stability condition.

\textit{Remark A2:}
We now consider a system with two time delays:
\begin{align}
	\label{p8a}
	&\dot{e}(t)=Ne(t)+N_{\tau}e(t-\tau)+N_h e(t-h),\\
	\label{p8b}
	&e(t)=\phi(t), \quad \forall t \in [-h, 0],
\end{align}
where \(e(t)\in \mathbb{R}^n\) is the state vector, and \(\phi(t)\) is the initial function. The scalars \(\tau>0\) and \(h>\tau\) represent time delays, and \(N, N_{\tau}, N_h \in \mathbb{R}^{n\times n}\) are constant matrices.

By setting
\(\tau_1=\frac{\tau}{2}\), \(\tau_2=\tau\), \(\tau_3=h\),
\(N_1=\mathbf{0}_{n\times n}\), \(N_2=N_{\tau}\), and \(N_3=N_h\), Lemma \ref{lem:Sta_multi} can be applied to derive a stability condition for system \eqref{p8a}. In the following, a numerical example is presented to illustrate the feasibility of the derived stability condition.

\textit{Example A3:}
Consider system \eqref{p8a} with matrices
\[
N=\begin{pmatrix}
	0 & 1\\
	-2 & 0.1
\end{pmatrix}, \quad
N_{\tau}=\begin{pmatrix}
	0 & 0\\
	1 & 0
\end{pmatrix}, \quad
N_{h}=\begin{pmatrix}
	0 & 0\\
	0.2 & 0
\end{pmatrix},
\]
and \(\tau=1.2\).

For \(\tau=1.2\), Remark A2 ensures that system  \eqref{p8a} is asymptotically stable for a maximum allowable time delay \(h = 1.68\).

\subsection{Stabilization of time-delay systems}	
In this section, based on the stability conditions derived in Section \ref{sec1.2.1}, we   investigate the stabilization problem for systems with one or multiple time delays. First, we look at time-delay system \eqref{p6a}. We consider the following two stabilization problems:

{\bf Problem 1:} Let $N\in \mathbb{R}^{n\times n}$ be a given matrix, which may be unstable. Our stabilization problem is to determine a matrix $N_{\tau}\in \mathbb{R}^{n\times n}$ such that the closed-loop system (\ref{p6a}) is asymptotically stable under the following cases:
\begin{itemize}
	\item{\it Case 1}: $\tau>0$ is a known constant delay; and		
	\item{\it Case 2}: $\tau>0$ is unknown but belongs to a known interval, i.e,  $\tau \in [\underline{\tau}, \overline{\tau}]$.
\end{itemize}

{\bf Problem 2:} Let \(N := N_{0,1} + ZN_{0,2} \in \mathbb{R}^{n\times n}\) and \(N_{\tau} := N_{\tau,1} + ZN_{\tau,2} \in \mathbb{R}^{n\times n}\), where \(N_{0,1}, N_{\tau,1}, \in \mathbb{R}^{n\times n}\) and \(N_{0,2},  N_{\tau,2} \in \mathbb{R}^{m\times n}\) are given matrices, and \(Z \in \mathbb{R}^{n\times m}\) is the matrix to be determined.

The stabilization problem is to determine a matrix \(Z \in \mathbb{R}^{n\times m}\) such that the closed-loop system \eqref{p6a} is asymptotically stable under the following cases:
\begin{itemize}
	\item{\it Case 1}: \(\tau > 0\) is a known constant delay; and
	\item{\it Case 2}: \(\tau > 0\) is unknown but belongs to a known interval, i.e., \(\tau \in [\underline{\tau}, \overline{\tau}]\).
\end{itemize}

In the following, based on the stability conditions derived in Section \ref{sec1.2.1}, we establish corresponding stabilization conditions for Problems 1 and 2.

For Problem 1 under Case 1, we employ Lemma \ref{lem:Sta_multi} together with Remark A1 to stabilize system \eqref{p6a}. Accordingly, we adopt the same notations as those introduced in Lemma \ref{lem:Sta_multi}. In addition, for the derivation of the stabilization condition of system \eqref{p6a}, we introduce the following two auxiliary notations:
\begin{align}\label{eq:math-N1}
	\mathcal{N}_1^{\mathsf T} :=
	\big(N \;\; \mathbf 0_{n\times n} \;\; \mathbf 0_{n\times n} \;\; \mathbf 0_{n\times n} \;\; \mathbf 0_{n\times 3n}\;\; -I_n\big) \in \mathbb{R}^{n\times 8n},
\end{align}
and
\begin{align}\label{eq:math-N2}
	\mathcal{N}_2^{\mathsf T} :=
	\big(\mathbf 0_{n\times n}\;\; \mathbf 0_{n\times n}\;\; \mathbf 0_{n\times n} \;\; I_n  \;\; \mathbf 0_{n\times 3n}\;\; \mathbf 0_{n\times n}\big) \in \mathbb{R}^{n\times 8n}.
\end{align}

By setting
\(\tau_1:=\frac{\tau}{3}\), \(\tau_2:=\frac{2\tau}{3}\) , \(\tau_3:=\tau\), and following the same arguments as in the proof of Lemma \ref{lem:Sta_multi}, we derive a stabilization condition for system \eqref{p6a}, as stated in the following lemma.
\begin{lemma}\label{lem:stabilization-constant-delay}
	Suppose there exist a $4n\times 4n$-matrix $P>0$, six  $n\times n$-matrices
	$Q_1> 0$, $Q_2> 0$, $Q_3> 0$, $R_1> 0$, $R_2> 0$, $R_3> 0$,
	a $n\times n $-nonsingular matrix $X$, a $n\times n $ matrix $G$, a scalar $\lambda$ such that the following LMI holds:
	\begin{align}\label{eq:LMI4}
		\Theta_3+\mathrm{sym}\big((v_1+\lambda v_8)X\mathcal{N}_1^{\mathsf T}\big)
		+\mathrm{sym}\big((v_1+\lambda v_8)G\mathcal{N}_2^{\mathsf T}\big)<0,
	\end{align}
	where $\Theta_3$ is given in \eqref{Theta3}. Then, the closed-loop system
	(\ref{p6a}) is asymptotically stable with the controller gain given by
	\begin{align}\label{eq:N-tau}
		N_{\tau}=X^{-1}G.
	\end{align}
\end{lemma}
\textit{Proof:} Setting
\(\tau_1:=\frac{\tau}{3}\), \(\tau_2:=\frac{2\tau}{3}\) , \(\tau_3:=\tau\),
and substituting \eqref{eq:N-tau} into system \eqref{p6a}, we obtain
\begin{align}\label{p6a_new}
	&\dot{e}(t)=Ne(t)+X^{-1}{G}e(t-\tau_3),
\end{align}
which implies that
\begin{align}\label{p6a_new-rewritten}
	0=({\mathcal N}_1^{\mathsf T}+X^{-1}{G}{\mathcal N}_2^{\mathsf T})\xi(t).
\end{align}
Multiplying \eqref{p6a_new-rewritten} on the left by $2\xi^{\mathsf T}(t)(v_1+\lambda v_8)X$, we obtain
\begin{align}\label{eq:descriptor4}
	0=2\xi^{\mathsf T}(t)\big(v_1+v_8\lambda \big)X\big(\mathcal{N}_1^{\mathsf T}+X^{-1}{G}\mathcal{N}_2^{\mathsf T}\big)\xi(t),
\end{align}	
which further yields
\begin{align}\label{eq:descriptor4-1}
	0=2\xi^{\mathsf T}(t)\Big(\big(v_1+\lambda v_8 \big)X\mathcal{N}_1^{\mathsf T}+\big(v_1+\lambda v_8 \big)G\mathcal{N}_2^{\mathsf T}\Big)\xi(t).
\end{align}	
By using \eqref{eq:descriptor4-1}  and the same arguments as in the proof of Lemma \ref{lem:Sta_multi}, we conclude that the closed-loop system \eqref{p6a} with $N_{\tau}$ in \eqref{eq:N-tau} is asymptotically stable. This completes the proof of Lemma \ref{lem:stabilization-constant-delay}.

Similarly, by employing Lemma \ref{lem:Sta_interval_delay}, we establish a stabilization condition for system \eqref{p6a} in Case 2, i.e. we solve Problem 1 for the case where the time-delay is unknown but belongs to a known interval. To this end, besides the notation used in deriving a stability condition in Lemma \ref{lem:Sta_interval_delay}, we introduce the following two notations:

\begin{align}\label{eq:math-N1-interval}
	\mathcal{N}_1^{\mathsf T} :=
	\big(N\;\; \mathbf 0_{n\times n} \;\; \mathbf 0_{n\times n} \;\; \mathbf 0_{n\times n} \;\; \mathbf 0_{n\times 3n}\;\; -I_n\big) \in \mathbb{R}^{n\times 8n},
\end{align}
and
\begin{align}\label{eq:math-N2-interval}
	\mathcal{N}_2^{\mathsf T} :=
	\big(\mathbf 0_{n\times n}\;\; \mathbf 0_{n\times n} \;\; I_n \;\; \mathbf 0_{n\times n} \;\; \mathbf 0_{n\times 3n}\;\; \mathbf 0_{n\times n}\big) \in \mathbb{R}^{n\times 8n}.
\end{align}

A stabilizability condition for system (\ref{p6a}) under Problem 1, Case 2 is stated in the following lemma.
\begin{lemma}\label{lem:Stabilization_interval_delay}
	Suppose there exist a $4n\times 4n$-matrix $P>0$, five  $n\times n$-matrices
	$Q_1> 0$, $Q_2> 0$, $Q_3> 0$, $R_1> 0$, $R_2> 0$,
	a $n\times n $-nonsingular matrix $X$, a $n\times n $ matrix $G$, a scalar $\lambda$ such that the following LMIs hold:	
	\begin{align}
		&\Theta_1(\tau)+\mathrm{sym}\big((v_1+\lambda v_8)X\mathcal{N}_1^{\mathsf T}\big)
		+\mathrm{sym}\big((v_1+\lambda v_8)G\mathcal{N}_2^{\mathsf T}\big)\nonumber\\
		&\hspace{1cm}<0,\  \tau \in \{\underline{\tau}, \overline{\tau}\}, \label{eq:LMI5}\\
		&\Xi(\tilde{R}_2,S)>0, \label{eq:LMI_Xi5}
	\end{align}
	where $\Theta_1(\tau)$ and $\Xi(\tilde{R}_2,S)$ are defined in \eqref{Theta_1} and \eqref{eq:LMI_Xi1}, respectively.  	
	Then, the closed-loop system (\ref{p6a}) is asymptotically stable for any delay \(\tau \in [\underline{\tau}, \overline{\tau}]\), with the controller gain given by
	\begin{align}\label{eq:N-tau2}
		N_{\tau}=X^{-1}G.
	\end{align}	
\end{lemma}
\textit{Proof:} Similar to the proof of Lemma \ref{lem:stabilization-constant-delay}, for the closed-loop system
\begin{align}\label{p6a_new2}
	&\dot{e}(t)=Ne(t)+X^{-1}{G}e(t-\tau),
\end{align}
we also obtain the equation \eqref{eq:descriptor4-1} where $\mathcal{N}_1^{\mathsf T},\mathcal{N}_2^{\mathsf T} $ are defined in \eqref{eq:math-N1-interval}, \eqref{eq:math-N2-interval}, respectively. By the same arguments as in the proof of Lemma \ref{lem:Sta_interval_delay}, we conclude that the closed-loop system \eqref{p6a}, with $N_{\tau}$ given in \eqref{eq:N-tau2}, is asymptotically stable  for any delay \(\tau \in [\underline{\tau}, \overline{\tau}]\). This completes the proof of Lemma
\ref{lem:Stabilization_interval_delay}.\\

For Problem 2 under Case 1, we replace the notations in \eqref{eq:math-N1} and \eqref{eq:math-N2} with the following ones:
\begin{align}\label{eq:math-N1-new}
	\mathcal{N}_1^{\mathsf T} :=
	\big(N_{0,1}\;\; \mathbf 0_{n\times n} \;\; \mathbf 0_{n\times n} \;\; N_{\tau,1} \;\; \mathbf 0_{n\times 3n}\;\; -I_n\big) \in \mathbb{R}^{n\times 8n},
\end{align}
and
\begin{align}\label{eq:math-N2-new}
	\mathcal{N}_2^{\mathsf T} :=
	\big(N_{0,2}\;\; \mathbf 0_{m\times n}\;\; \mathbf 0_{m\times n} \;\; N_{\tau,2}  \;\; \mathbf 0_{m\times 3n}\;\; \mathbf 0_{m\times n}\big) \in \mathbb{R}^{m\times 8n}.
\end{align}
Then, by employing the free-weighting matrix technique, it follows that, for any $n\times n $-matrix $X$, and any $n\times m $-matrix $G$,
\begin{align}
	0 = {} & \xi_0^{\mathsf T}(t) \Bigl( \mathrm{sym} \bigl( (v_1 + \lambda v_8) X \mathcal{N}_1^{\mathsf T} \bigr) \nonumber \\
	& + \mathrm{sym} \bigl( (v_1 + \lambda v_8) G \mathcal{N}_2^{\mathsf T} \bigr) \Bigr) \xi_0(t). \label{eq:free-weighting-matrix2}
\end{align}
Following the same line of reasoning as in Lemma \ref{lem:stabilization-constant-delay}, we obtain a stabilizability condition for system \eqref{p6a} in Problem 2, Case 1 and it is stated in the following Lemma.

\begin{lemma}\label{lem:output-stabilization-constant-delay}
	Suppose there exist a $4n\times 4n$-matrix $P>0$, six  $n\times n$-matrices
	$Q_1> 0$, $Q_2> 0$, $Q_3> 0$, $R_1> 0$, $R_2> 0$, $R_3> 0$,
	a $n\times n $-nonsingular matrix $X$, a $n\times m $-matrix $G$, and a scalar $\lambda$ such that the following LMI holds:
	\begin{align}\label{eq:LMI6}
		\Theta_3+\mathrm{sym}\big((v_1+\lambda v_8)X\mathcal{N}_1^{\mathsf T}\big)
		+\mathrm{sym}\big((v_1+\lambda v_8)G\mathcal{N}_2^{\mathsf T}\big)<0,
	\end{align}
	where $\Theta_3$ is defined in \eqref{Theta3}. Then, the closed-loop system (\ref{p6a}) is asymptotically stable with the controller gain given by
	\begin{align}\label{eq:ZZ1}
		Z=X^{-1}G.
	\end{align}		
\end{lemma}

For Problem 2 under Case 2, we introduce the following notations:
\begin{align}\label{eq:math-N1-interval2}
	\mathcal{N}_1^{\mathsf T} :=
	\big(N_{0,1}\;\; \mathbf 0_{n\times n}  \;\;
	N_{\tau,1}\;\; \mathbf 0_{n\times n} \;\; \mathbf 0_{n\times 3n}\;\; -I_n\big) \in \mathbb{R}^{n\times 8n},
\end{align}
and
\begin{align}\label{eq:math-N2-interval2}
	\mathcal{N}_2^{\mathsf T} :=
	\big(N_{0,2}\;\; \mathbf 0_{m\times n} \;\;
	N_{\tau,2}\;\;  \mathbf 0_{m\times n} \;\; \mathbf 0_{m\times 3n}\;\; \mathbf 0_{m\times n}\big) \in \mathbb{R}^{m\times 8n}.
\end{align}
Similar to the free-weighting-matrix technique used in Lemma \ref{lem:output-stabilization-constant-delay}, we obtain the equation \eqref{eq:free-weighting-matrix2} for $\mathcal{N}_1^{\mathsf T},\mathcal{N}_2^{\mathsf T} $ defined in \eqref{eq:math-N1-interval2}, \eqref{eq:math-N2-interval2}, respectively. Combining this result with  Lemma \ref{lem:Sta_interval_delay}, we derive a stabilizability condition for system \eqref{p6a} in Problem 2, Case 2, as stated in the following lemma.
\begin{lemma}\label{lem:output-Stabilization_interval_delay}
	Suppose that there exist a ${4n\times 4n}$ matrix \(P>0\), five ${n\times n}$ matrices $Q_1>0$, $Q_2>0$, $Q_3>0$, $R_1>0$, $R_2>0$, a nonsingular matrix \(X\in \mathbb{R}^{n\times n}\), a matrix \(G \in \mathbb{R}^{n\times m}\), a ${2n\times 2n}$ matrix \(S\), and a scalar $\lambda$
	such that the following LMIs hold:
	\begin{align}
		&\Theta_1(\tau) + \mathrm{sym}\big((v_1 + \lambda v_8)X\mathcal{N}_1^{\mathsf T}\big) \nonumber \\
		&\quad + \mathrm{sym}\big((v_1 + \lambda v_8)G\mathcal{N}_2^{\mathsf T}\big) < 0, \quad \tau \in \{\underline{\tau}, \overline{\tau}\}, \label{eq:nLMI7} \\
		&\Xi(\tilde{R}_2, S) > 0, \label{eq:LMI_Xi7}
	\end{align}
	where $\Theta_1(\tau)$ and $\Xi(\tilde{R}_2,S)$ are defined in \eqref{Theta_1} and \eqref{eq:LMI_Xi1}, respectively. Then,
	the closed-loop system (\ref{p6a}) is asymptotically stable for any delay \(\tau \in [\underline{\tau}, \overline{\tau}]\), with the controller gain given by
	\begin{align}\label{eq:ZZ2}
		Z=X^{-1}G.
	\end{align}	
\end{lemma}

\textit{Example A4:} Consider Problem 1 for system \eqref{p6a} with the matrix $$N=\begin{pmatrix}
	0.2 & 0  &  0\\
	0.2 &  0.1 & -0.1\\
	0   & 0.2 & 0.15
\end{pmatrix}.$$ Noting that, in this example, the matrix $N$ is unstable.

For Case 1, by applying Lemma \ref{lem:stabilization-constant-delay}, we  obtain a maximum allowable delay $\tau = 4.8$ and the matrix  $$N_{\tau}=\begin{pmatrix}
	-0.2033  &  0.0001  & -0.0004\\
	-0.1619 &  -0.1403  &  0.0839\\
	0.0195 &  -0.1707 &  -0.1751
\end{pmatrix},$$
which guarantees that the closed-loop system  is asymptotically stable.

For Case 2, by applying Lemma \ref{lem:Stabilization_interval_delay}, we obtain an allowable delay interval $\tau\in [2, 4.78] $ and the matrix  $$N_{\tau}=\begin{pmatrix}
	-0.2066  &  0.0004 &   0.0003\\
	-0.1570  & -0.1460  &  0.0833\\
	0.0174  & -0.1751  & -0.1731
\end{pmatrix},$$
which ensures that the closed-loop system  is asymptotically stable for
any time-delay $\tau\in [2, 4.78].$

\textit{Example A5:}  Consider Problem 2 for system \eqref{p6a} with the matrices
\begin{gather*}
	N_{0,1} = \begin{pmatrix}
		0.2 & 0  &  0 \\
		0.2 &  0.1 & -0.1 \\
		0   & 0.2 & 0.15
	\end{pmatrix}, \quad
	N_{0,2} = \begin{pmatrix} 0 & 0 & 0 \end{pmatrix}, \\
	N_{\tau,1} = \mathbf{0}_{3\times 3}, \quad	
	N_{\tau,2} = \begin{pmatrix} 1 & 2 & 3 \end{pmatrix}.
\end{gather*}
We aim to determine a matrix $Z\in \mathbb{R}^{3\times 1}$ such that the closed-loop system
\begin{align}\label{p6a_closed}
	&\dot{e}(t)=(N_{0,1} + ZN_{0,2})e(t)+(N_{\tau,1} + ZN_{\tau,2})e(t-\tau),
\end{align}
is asymptotically stable. \\

For Case 1, by applying Lemma \ref{lem:output-stabilization-constant-delay}, we obtain a maximum allowable delay $\tau = 2.2$ and the matrix  $$Z=\begin{pmatrix}
	-0.0568\\ 	-0.0726\\   	-0.0601
\end{pmatrix},$$
which guarantees that the closed-loop system \eqref{p6a_closed}  is asymptotically stable.

For Case 2, by applying Lemma \ref{lem:output-Stabilization_interval_delay}, we obtain an allowable delay interval $\tau\in [1, 2.1] $ and the matrix
$$Z=\begin{pmatrix}
	-0.0646 \\	-0.0860 \\	-0.0612
\end{pmatrix},$$
which ensures that the closed-loop system \eqref{p6a_closed}  is asymptotically stable for any delay $\tau\in [1, 2.1].$	

It can be observed that, the solution to Problem 1 yields a larger allowable delay bound or a wider range of allowable delays compared with the solution to Problem 2.

Next, we consider the output-feedback stabilization problem for system \eqref{p8a} with two constant time delays $0<\tau<h$. To this end, let
\begin{align*} 
	N &:= N_{0,1} + Z_0N_{0,2}, \quad N_{\tau} := N_{\tau,1} + Z_{\tau}N_{\tau,2}, \\
	N_{h} &:= N_{h,1} + Z_{h}N_{h,2}.
\end{align*}
where \(N_{0,1}, N_{\tau,1}, N_{h,1},  \in \mathbb{R}^{n\times n}\), \(N_{0,2}, N_{\tau,2}, N_{h,2}  \in \mathbb{R}^{m\times n}\) are given matrices, and \(Z_0, Z_{\tau}$, $Z_{h} \in \mathbb{R}^{n\times m}\) are matrices to be designed.

{\bf Problem 3}:  Determine three matrices \(Z_0, Z_{\tau}, Z_{h} \in \mathbb{R}^{n\times m}\) such that the closed-loop system \eqref{p8a} is asymptotically stable.

We introduce the following notations:
\begin{align}
	\mathcal{N}_1^{\mathsf T} &:= \big(N_{0,1} \, \mathbf{0}_{n\times n} \, N_{\tau,1} \, N_{h,1} \, \mathbf{0}_{n\times 3n} \, -I_n\big) \in \mathbb{R}^{n\times 8n}, \label{eq:math-N1-two-delay} \\
	\mathcal{N}_{2,0}^{\mathsf T} &:= \big(N_{0,2} \, \mathbf{0}_{m\times n} \, \mathbf{0}_{m\times n} \, \mathbf{0}_{m\times n} \, \mathbf{0}_{m\times 3n} \, \mathbf{0}_{m\times n}\big) \nonumber \\
	&\quad \in \mathbb{R}^{m\times 8n}, \label{eq:math-N2-0} \\
	\mathcal{N}_{2,\tau}^{\mathsf T} &:= \big(\mathbf{0}_{m\times n} \, \mathbf{0}_{m\times n} \, N_{\tau,2} \, \mathbf{0}_{m\times n} \, \mathbf{0}_{m\times n} \, \mathbf{0}_{m\times 3n} \, \mathbf{0}_{m\times n}\big) \nonumber \\
	&\quad \in \mathbb{R}^{m\times 8n}, \label{eq:math-N2-tau} \\
	\mathcal{N}_{2,h}^{\mathsf T} &:= \big(\mathbf{0}_{m\times n} \, \mathbf{0}_{m\times n} \, \mathbf{0}_{m\times n} \, N_{h,2} \, \mathbf{0}_{m\times 3n} \, \mathbf{0}_{m\times n}\big) \nonumber \\
	&\quad \in \mathbb{R}^{m\times 8n}. \label{eq:math-N2-h}
\end{align}
By setting \(\tau_1=\frac{\tau}{2}\), \(\tau_2=\tau\), \(\tau_3=h\)  and employing Lemma \ref{lem:Sta_multi}, Remark A2, and the free-weighting matrix technique, we derive an output-feedback stabilization condition for system \eqref{p8a}, as stated in the following lemma.
\begin{lemma}\label{lem:output-Stabilization-two-delay}
	Suppose there exist a $4n\times 4n$-matrix $P>0$, six  $n\times n$-matrices
	$Q_1> 0$, $Q_2> 0$, $Q_3> 0$, $R_1> 0$, $R_2> 0$, $R_3> 0$,
	a nonsingular matrix \(X\in \mathbb{R}^{n\times n}\), three matrices \(G_0, G_{\tau}, G_h \in \mathbb{R}^{n\times m}\), and a scalar $\lambda$ such that the following LMI holds:
	\begin{align}\label{eq:LMI7}
		&\Theta_3+\mathrm{sym}\big((v_1+\lambda v_8)X\mathcal{N}_1^{\mathsf T}\big)
		+\mathrm{sym}\big((v_1+\lambda v_8)G_0\mathcal{N}_{2,0}^{\mathsf T}\big)\nonumber\\
		&	+\mathrm{sym}\big((v_1+\lambda v_8)G_{\tau}\mathcal{N}_{2,\tau}^{\mathsf T}
		+\mathrm{sym}\big((v_1+\lambda v_8)G_{h}\mathcal{N}_{2,h}^{\mathsf T}\nonumber\\
		&<0,
	\end{align}
	where $\Theta_3$ is defined in \eqref{Theta3}. Then, the closed-loop system \eqref{p8a} is asymptotically stable with the controller gains given by
	\begin{align}\label{eq:Z-1}
		Z_0=X^{-1}G_0, \ Z_{\tau}=X^{-1}G_{\tau},\ Z_h=X^{-1}G_h.
	\end{align}	
\end{lemma}

\textit{Example A6:}
Consider Problem 3 for system \eqref{p8a} with the matrices
\[
N_{0,1}=2,\quad N_{0,2}=0,\quad
N_{\tau,1}=0,\quad N_{\tau,2}=1,\quad
N_{h,1}=0,\]
$N_{h,2}=1$
and \(0<\tau<h\).

We aim to determine \(Z_0, Z_{\tau}, Z_{h} \in \mathbb{R}^{1\times 1}\), the largest possible time delay \(\tau\), and a time delay \(h>\tau\) such that the closed-loop system \eqref{p8a} is asymptotically stable.

By Lemma \ref{lem:output-Stabilization-two-delay}, we obtain \(\tau=0.595\), \(h=0.8\), \(Z_0=0\), \(Z_{\tau}=-3.1605\), and \(Z_h=1.1556\), which ensure that the closed-loop system \eqref{p8a} is asymptotically stable.

Notably, in the case where \(N_{h,2}=0\), i.e., when stabilizing system \eqref{p8a} using only a delayed feedback term, Lemma \ref{lem:output-stabilization-constant-delay} yields a smaller allowable time delay \(\tau=0.495\). This demonstrates that employing two delayed feedback terms can increase the admissible time delay for system stabilization.

\textit{Example A7:}
Consider Problem 3 for system \eqref{p8a} with the matrices
\[
N_{0,1}=\begin{pmatrix}
	0.2 & 0   & 0\\
	0.2 & 0.1 & -0.1\\
	0   & 0.2 & 0.15
\end{pmatrix}, \quad
N_{0,2}=\begin{pmatrix}
	0 & 0 & 0
\end{pmatrix},
\]
\[
N_{\tau,1}=N_{h,1}=\mathbf{0}_{3\times 3}, \quad
N_{\tau,2}=N_{h,2}=\begin{pmatrix}1 & 2 & 3\end{pmatrix},
\]
and \(0<\tau<h\). We aim to determine matrices \(Z_0, Z_{\tau}, Z_{h} \in \mathbb{R}^{3\times 1}\), the largest possible time delay \(\tau\), and a time delay \(h>\tau\) such that the closed-loop system \eqref{p8a} is asymptotically stable.

By Lemma \ref{lem:output-Stabilization-two-delay}, we obtain \(\tau=2.43\), \(h=2.7\), and
\[
Z_0=\begin{pmatrix}
	0\\
	0\\
	0
\end{pmatrix}, \quad
Z_{\tau}=\begin{pmatrix}
	-0.2184\\
	-0.2402\\
	-0.1099
\end{pmatrix}, \quad
Z_h=\begin{pmatrix}
	0.1554\\
	0.1633\\
	0.0524
\end{pmatrix},
\]
which ensure that the closed-loop system \eqref{p8a} is asymptotically stable.

Notably, in the case where \(N_{h,2}=\mathbf{0}_{1\times 3}\), i.e., when stabilizing system \eqref{p8a} using only a delayed feedback term, Lemma \ref{lem:output-stabilization-constant-delay} yields a smaller allowable time delay \(\tau=2.2\) (see Example A5). Again, this demonstrates that employing two delayed feedback terms can increase the admissible time delay for system stabilization.

Lastly, we consider the output-feedback stabilization problem for system \eqref{p7a} with three constant time delays satisfying \(0<\tau_1<\tau_2<\tau_3\). To this end, let
\[
N := N_{0,1} + Z_0N_{0,2}, \quad
N_1 := N_{1,1} + Z_1N_{1,2},
\]
\[
N_2 := N_{2,1} + Z_2N_{2,2}, \quad
N_3 := N_{3,1} + Z_3N_{3,2},
\]
where

\(N_{0,1}, N_{1,1}, N_{2,1}, N_{3,1} \in \mathbb{R}^{n\times n}\),
\(N_{0,2}, N_{1,2}, N_{2,2}, N_{3,2} \in \mathbb{R}^{m\times n}\) are given matrices,
and \(Z_0, Z_1, Z_2, Z_3 \in \mathbb{R}^{n\times m}\) are matrices to be designed.\\

\noindent\textbf{Problem 4:} Determine matrices \(Z_0, Z_1, Z_2, Z_3 \in \mathbb{R}^{n\times m}\) such that the closed-loop system \eqref{p7a} is asymptotically stable.

We introduce the following notations:
\begin{align}
	\mathcal{N}_1^{\mathsf T} &:= \big(N_{0,1} \, N_{1,1} \, N_{2,1} \, N_{3,1} \, \mathbf{0}_{n\times 3n} \, -I_n\big) \in \mathbb{R}^{n\times 8n}, \label{eq:math-N1-three-delay} \\
	\mathcal{N}_{2,0}^{\mathsf T} &:= \big(N_{0,2} \, \mathbf{0}_{m\times n} \, \mathbf{0}_{m\times n} \, \mathbf{0}_{m\times n} \, \mathbf{0}_{m\times 3n} \, \mathbf{0}_{m\times n}\big) \nonumber \\
	&\quad \in \mathbb{R}^{m\times 8n}, \label{eq:math-N2-0} \\
	\mathcal{N}_{2,1}^{\mathsf T} &:= \big(\mathbf{0}_{m\times n} \, N_{1,2} \, \mathbf{0}_{m\times n} \, \mathbf{0}_{m\times n} \, \mathbf{0}_{m\times 3n} \, \mathbf{0}_{m\times n}\big) \nonumber \\
	&\quad \in \mathbb{R}^{m\times 8n}, \label{eq:math-N2-1} \\
	\mathcal{N}_{2,2}^{\mathsf T} &:= \big(\mathbf{0}_{m\times n} \, \mathbf{0}_{m\times n} \, N_{2,2} \, \mathbf{0}_{m\times n} \, \mathbf{0}_{m\times 3n} \, \mathbf{0}_{m\times n}\big) \nonumber \\
	&\quad \in \mathbb{R}^{m\times 8n}, \label{eq:math-N2-2} \\
	\mathcal{N}_{2,3}^{\mathsf T} &:= \big(\mathbf{0}_{m\times n} \, \mathbf{0}_{m\times n} \, \mathbf{0}_{m\times n} \, N_{3,2} \, \mathbf{0}_{m\times 3n} \, \mathbf{0}_{m\times n}\big) \nonumber \\
	&\quad \in \mathbb{R}^{m\times 8n}. \label{eq:math-N2-3}
\end{align}

By employing Lemma \ref{lem:Sta_multi} together with the free-weighting matrix technique, we derive an output-feedback stabilization condition for system \eqref{p7a}, as stated in the following lemma.

\begin{lemma}\label{lem:output-Stabilization-three-delay}
	Suppose there exist a \(4n\times 4n\) matrix \(P>0\), six \(n\times n\) matrices
	\(Q_1> 0\), \(Q_2> 0\), \(Q_3> 0\), \(R_1> 0\), \(R_2> 0\), \(R_3> 0\),
	a matrix \(X\in \mathbb{R}^{n\times n}\), and four matrices \(G_0, G_{1}, G_2, G_3 \in \mathbb{R}^{n\times m}\),
	and a scalar $\lambda$ such that the following LMI holds:
	\begin{align}\label{eq:LMI8}
		\Theta_3
		&+\mathrm{sym}\big((v_1+\lambda v_8)X\mathcal{N}_1^{\mathsf T}\big)
		+\mathrm{sym}\big((v_1+\lambda v_8)G_0\mathcal{N}_{2,0}^{\mathsf T}\big)\nonumber\\
		&+\mathrm{sym}\big((v_1+\lambda v_8)G_{1}\mathcal{N}_{2,1}^{\mathsf T}\big)
		+\mathrm{sym}\big((v_1+\lambda v_8)G_{2}\mathcal{N}_{2,2}^{\mathsf T}\big)\nonumber\\
		&+\mathrm{sym}\big((v_1+\lambda v_8)G_{3}\mathcal{N}_{2,3}^{\mathsf T}\big)
		< 0,
	\end{align}
	where \(\Theta_3\) is defined in \eqref{Theta3}. Then, the closed-loop system \eqref{p7a} is asymptotically stable with the controller gains given by	
	\begin{align}
		Z_0 &= X^{-1}G_0, \quad Z_1 = X^{-1}G_1, \nonumber \\
		Z_2 &= X^{-1}G_2, \quad Z_3 = X^{-1}G_3. \label{eq:Z-1}
	\end{align}
\end{lemma}

\textit{Example A8:} 
Consider Problem 4 for system \eqref{p7a} where
$
N := N_{0,1} + Z_0N_{0,2}, $
$	N_1 := N_{1,1} + Z_1N_{1,2},$ \
$N_2 := N_{2,1} + Z_2N_{2,2},$ \
$	N_3 := N_{3,1} + Z_3N_{3,2},
$ with the matrices

\begin{align*}
	N_{0,1} &= \begin{pmatrix}
		0.2 & 0   & 0 \\
		0.2 & 0.1 & -0.1 \\
		0   & 0.2 & 0.15
	\end{pmatrix}, \, N_{1,1} = N_{2,1} = N_{3,1} = \mathbf{0}_{3\times 3}, \\
	N_{0,2} &= \begin{pmatrix} 0 & 0 & 0 \end{pmatrix}, \, N_{1,2} = N_{2,2} = \begin{pmatrix} 1 & 2 & 3 \end{pmatrix}, \\
	N_{3,2} &= \begin{pmatrix} 1 & 0 & 0 \end{pmatrix}.
\end{align*}
and \(\tau_1=3.65,\ \tau_2=3.7,\ \tau_3=3.75\).

We aim to determine matrices \(Z_0, Z_{1}, Z_{2}, Z_3 \in \mathbb{R}^{3\times 1}\), such that the closed-loop system \eqref{p7a} is asymptotically stable.

By Lemma \ref{lem:output-Stabilization-three-delay}, we obtain
\begin{align*}
	Z_0 &= \begin{pmatrix} 0 \\ 0 \\ 0 \end{pmatrix}, \quad
	Z_1 = \begin{pmatrix} -0.1159 \\ -0.3993 \\ -0.9907 \end{pmatrix}, \\
	Z_2 &= \begin{pmatrix} 0.1158 \\ 0.3869 \\ 0.9271 \end{pmatrix}, \quad
	Z_3 = \begin{pmatrix} -0.2249 \\ -0.1297 \\ 0.0445 \end{pmatrix}.
\end{align*}
which ensure that the closed-loop system \eqref{p7a} is asymptotically stable.

Compared with Example A7, where system \eqref{p7a} is stabilized with $\tau = 2.43$, the present example shows that system \eqref{p7a} remains stable for a larger delay of $\tau_1 = 3.65$. This indicates that, compared to using only two delayed feedback terms, employing three delayed feedback terms can enlarge the admissible delay bound for system stabilization.

\end{document}